\documentclass[11pt,a4paper]{article}
\usepackage{fullpage,graphicx,amssymb,amsmath,xspace,amsfonts}

\newcommand{\vol}{\mathrm{volume}}
\newcommand{\area}{\mathrm{area}}
\newcommand{\Reals}{\mathbb{R}}
\newcommand{\Integers}{\mathbb{Z}}

\newcommand{\unittile}{\ensuremath{U}\xspace}
\newcommand{\prefix}{\mathrm{pre}}
\newcommand{\postfix}{\mathrm{post}}
\newcommand{\plmin}{\pm}
\newcommand{\degr}{\mathrm{degr}}
\newcommand{\tdegr}{\mathrm{tiles}}
\newcommand{\doors}{\mathrm{ends}}
\newcommand{\angl}{\mathrm{angle}}
\newcommand{\turn}{\mathrm{turn}}
\newcommand{\Angl}{\mathrm{Angle}}
\newcommand{\Turn}{\mathrm{Turn}}
\newcommand{\ARRWW}{$AR^2W^2$\xspace}

\def\andfrac#1/#2{%
   \leavevmode\kern.1em
   \raise.5ex\hbox{\the\scriptfont0 #1}\kern-.1em
   /\kern-.15em\lower.25ex\hbox{\the\scriptfont0 #2}}

\newtheorem{theorem}{Theorem}
\newtheorem{lemma}{Lemma}
\newenvironment{proof}{Proof:}{\qed}
\def\squareforqed{\hbox{\rlap{$\sqcap$}$\sqcup$}}
\def\qed{\ifmmode\squareforqed\else{\unskip\nobreak\hfil
\penalty50\hskip1em\null\nobreak\hfil\squareforqed
\parfillskip=0pt\finalhyphendemerits=0\endgraf}\fi}

\newtheorem{observation}{Observation}

\begin{document}

\title{Recursive tilings and space-filling curves with little fragmentation}

\author{%
Herman~Haverkort\thanks{Dept.\ of Computer Science, Eindhoven University of Technology, the Netherlands, cs.herman@haverkort.net}
}
\maketitle

\begin{abstract}
This paper defines the Arrwwid number of a recursive tiling (or space-filling curve) as the smallest number $a$ such that any ball $Q$ can be covered by $a$ tiles (or curve sections) with total volume $O(\vol(Q))$. Recursive tilings and space-filling curves with low Arrwwid numbers can be applied to optimise disk, memory or server access patterns when processing sets of points in $\Reals^d$. This paper presents recursive tilings and space-filling curves with optimal Arrwwid numbers. For $d \geq 3$, we see that regular cube tilings and space-filling curves cannot have optimal Arrwwid number, and we see how to construct alternatives with better Arrwwid numbers.
\end{abstract}

\section{Introduction}
\subsection{The problem}
Consider a set of data points in a bounded region \unittile of $\Reals^2$, stored on disk. A standard operation on such point sets is to retrieve all points that lie inside a certain query range, for example a circle or a square. To prevent large delays because of disk head movements while answering such queries, it is desirable that the points are stored on disk in a clustered way~\cite{Asano1997,Jagadish1990,Jagadish1997,Kumar1994,Moon2001}. Similar considerations arise when storing spatial data in certain types of distributed networks~\cite{Scholl2009} or when scanning spatial objects to render them as a raster image; in the latter case it is desirable that the pixels that cover any particular object are scanned in a clustered way, so that the object does not have to be brought into cache too often~\cite{Voorhies1991}. For ease of explanation, we focus on the application of clustering to storing points on disks.

\begin{figure}
\centering
\includegraphics[height=7cm]{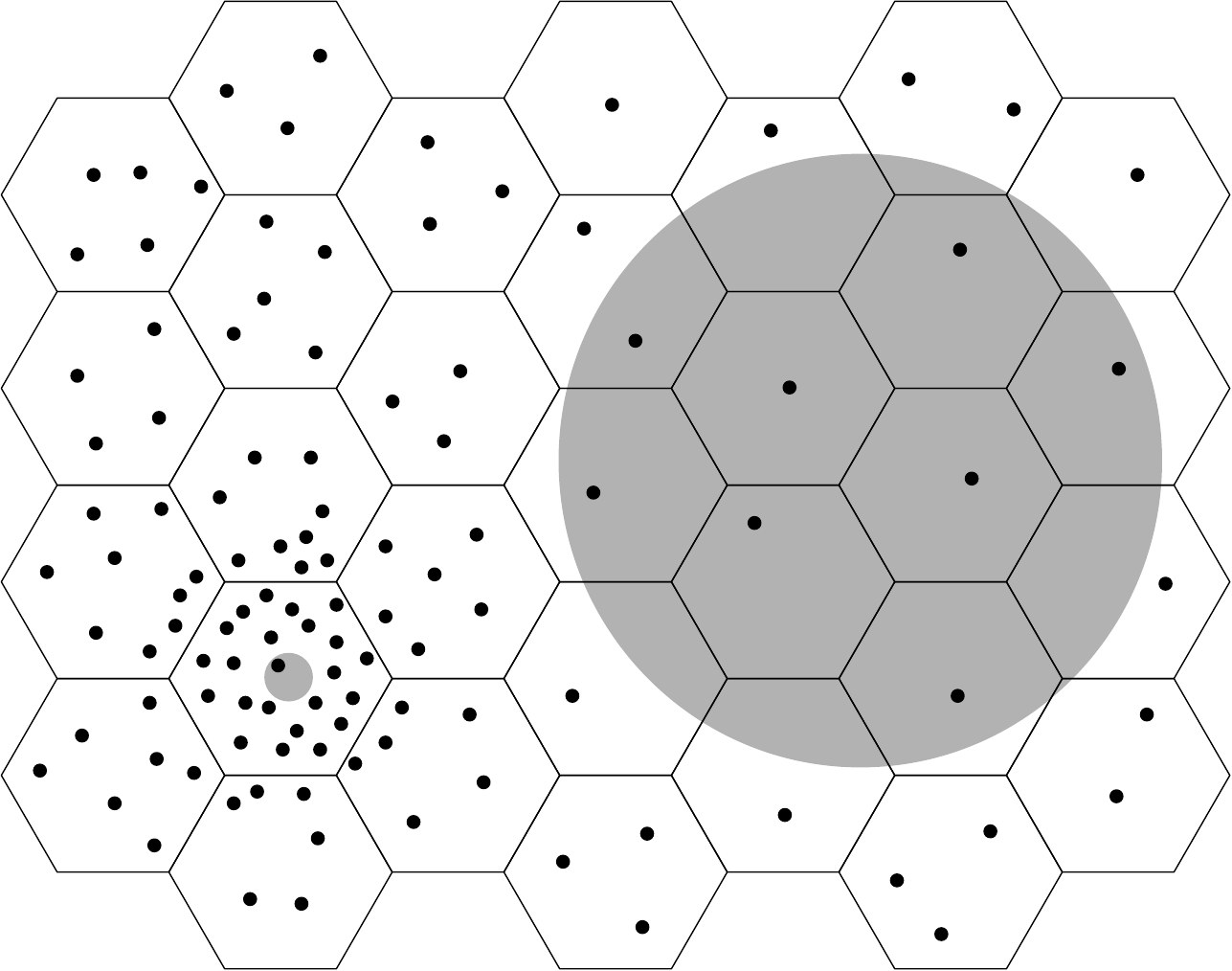}
\caption{Sorting points into (non-)recursive tiles. Queries with small query ranges (shaded disk on the left) may necessitate scanning the full contents of a crowded tile. Queries with large ranges (shaded disk on the right) may necessitate looking up many tiles.}
\label{fig:queries}
\end{figure}

We could try to achieve a good clustering in the following way. We divide \unittile into \emph{tiles}. The tiles could, for example, form a regular grid of hexagons (see Figure~\ref{fig:queries}). Now we store the points in each tile as a contiguous block on disk. To answer a query, say with a region $Q$ bounded by a circle, we compute which tiles intersect~$Q$. For each intersecting tile, we move the disk read head to the position where the first point in that tile is stored, and then we retrieve all points in the tile by just scanning them sequentially without further delays from disk head movements. Since some of the tiles that intersect $Q$ may lie partially outside~$Q$, some of the points thus retrieved may be \emph{false answers}: they are no answers to our query and need to be filtered out in post-processing.

The approach sketched above may work pretty well, provided the following conditions are met:\begin{enumerate}
\item We can compute efficiently which tiles intersect~$Q$.
\item We can figure out efficiently where the contents of these tiles are stored on disk (for example, using a small index to be kept in main memory).
\item A couple of tiles suffice to cover $Q$ (so that we do not have to move the disk read head to the starting point of another tile too often);
\item The tiles that cover $Q$ are not much larger than $Q$ (so that they do not contain too many false answers).
\end{enumerate}
These conditions are subject to a trade-off. On one extreme, we could use only one tile that covers all of \unittile (making the first three conditions trivial), bring the disk head into position only once, and scan all data points---including lots of false answers. On the other extreme, we could make the tiles so small that each of them contains very few points and we get very few false answers, but then we may get large delays because of disk head movements when going from tile to tile; the first and second condition may also be harder to satisfy in this case. The best balance is to be found somewhere between these extremes, but where?

If the data points are not uniformly distributed in \unittile and we choose any fixed tile size, then some of the tiles may contain many more data points than others. When the query range $Q$ is a small region inside such a crowded tile, we would need to retrieve all points in the tile, including many false answers (see Figure~\ref{fig:queries}). At the same time, there may be regions with many tiles with very few data points each. Queries that cover many such tiles (see Figure~\ref{fig:queries}) may incur an overhead for retrieving many separate tiles that is disproportionally large relative to the number of data points found.

\begin{figure}
\centering
\includegraphics[height=7cm]{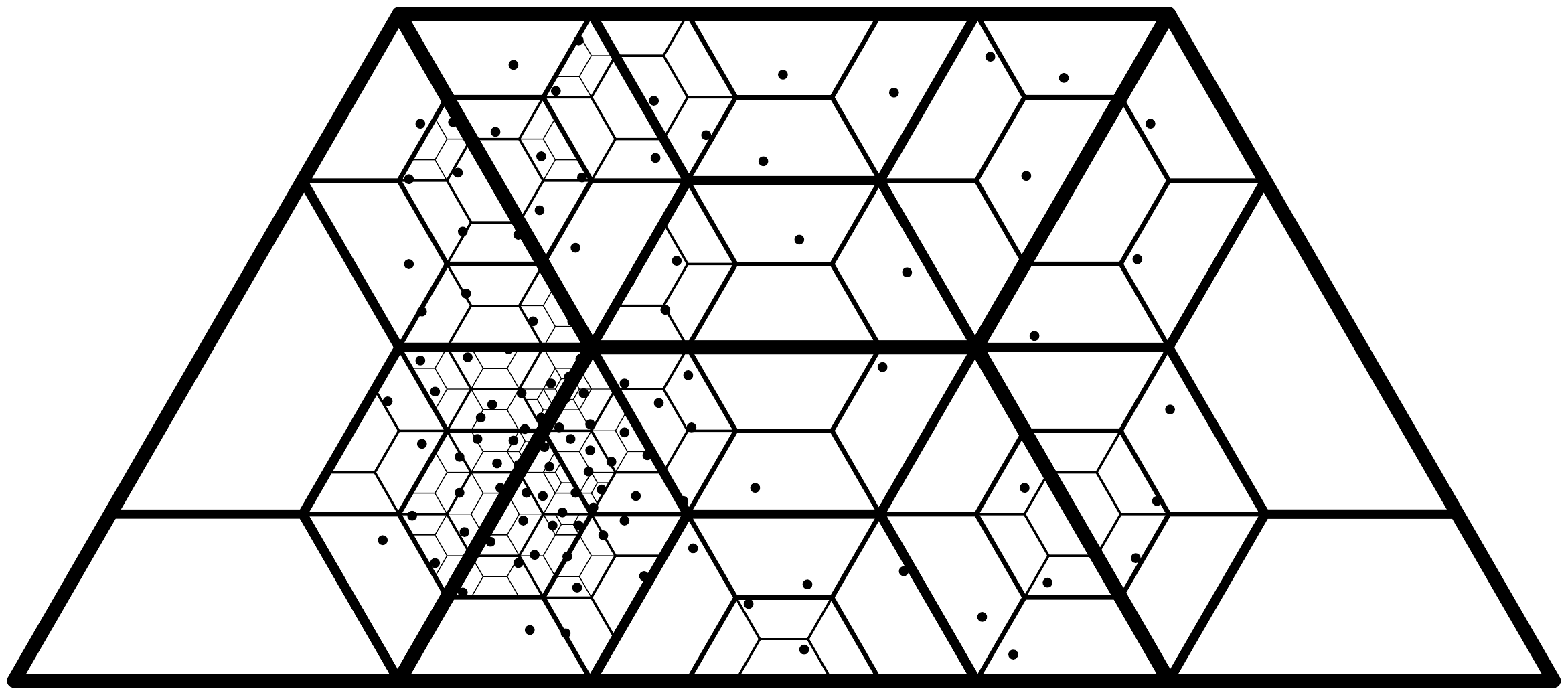}
\caption{A recursive tiling into trapezoids.}
\label{fig:trapezoids}
\end{figure}

This brings us to the topic of this paper: \emph{recursive} tilings. A recursive tiling is a subdivision of \unittile into tiles, that are each subdivided into tiles recursively. For example, Figure~\ref{fig:trapezoids} shows a recursive tiling based on trapezoids, expanded down to the level where each tile contains at most one data point. We store the data points in such a way that for each tile, on any level of recursion, the data points within that tile are stored as a contiguous sequence on the disk. Hopefully, if we get it right, we can now cover the region within any circle $Q$ with a small set of tiles from the level of recursion where the tiles have size proportional to $Q$. Thus we would satisfy condition three and four at the same time, and because we need to retrieve only few tiles, we could also hope to satisfy condition one and two. Considering that moving the disk head once may easily cost as much time as scanning ten thousands of points from disk, it is really important to keep the number of tiles small. Hence the topic of this study: what recursive tilings make it possible to cover any disk-shaped query region $Q$ with the smallest possible number of tiles, while still making sure that the total size of the tiles is at most some constant times the size of~$Q$?

In this paper, we study this question for two- and three-dimensional settings, and we give some basic results for higher dimensions. We also study what can be achieved by controlling the order in which tiles are stored. Note that until now, we only required that the data points within each tile are stored contiguously, but we did not assume anything about the order in which the subtiles of any tile are stored with respect to each other. A well-chosen order may have the result that some of the tiles used to cover the query range are stored consecutively on disk, thus eliminating the need to move the disk head when going from one tile to the next.

\subsection{The Arrwwid number}
To be able to make our problem statement more precise, we define the Arrwwid number of a recursive tiling of a region \unittile as follows:\begin{quote}
The Arrwwid number is the smallest number $a$ such that there is a constant $c$ such that any disk $Q$ that lies entirely in \unittile can be covered by $a$ tiles with total area at most $c \cdot \area(Q)$.
\end{quote}
Less formally: the Arrwwid number is the smallest number $a$ such that any disk is covered by at most $a$ relatively small tiles. We call $c$ the \emph{cover ratio}. The definition is essentially from Asano, Ranjan, Roos, Welzl and Widmayer~\cite{Asano1997}, and we name it \emph{Arrwwid number} in their honour.

When a recursive tiling is enhanced with a recursive definition of how the subtiles within each tile are ordered relative to each other, the result constitutes the definition of a \emph{recursive scanning order} or a \emph{recursive space-filling curve} (this paper uses these two terms interchangeably; they are explained in more detail in Section~\ref{sec:2DSFC-definition}). The Arrwwid number of a space-filling curve that fills a region \unittile is defined as follows:\begin{quote}
The Arrwwid number is the smallest number $a$ such that there is a constant $c$ such that any disk $Q$ that lies entirely in \unittile can be covered by $a$ curve fragments (that is, sets of consecutive tiles) with total area at most $c \cdot \area(Q)$.
\end{quote}

In the above definitions we could exchange disks for squares: this would only affect the cover ratios $c$ but not the Arrwwid numbers $a$. To see this, observe that if the Arrwwid number is $a$ with respect to disks, with cover ratio $c$, then any square $Q$ can be covered by a disk of area $\frac\pi2 \cdot \area(Q)$ which can subsequently be covered by at most $a$ tiles of total area at most $c \cdot \frac\pi2 \cdot \area(Q)$. Therefore the Arrwwid number with respect to squares is at most the Arrwwid number with respect to disks. Furthermore, if the Arrwwid number is $a$ with respect to squares, with cover ratio $c$, then any disk $Q$ can be covered by a square of area $\frac4\pi \cdot \area(Q)$, which can subsequently be covered by at most $a$ tiles of total area at most $c \cdot \frac4\pi \cdot \area(Q)$. Therefore the Arrwwid number with respect to disks is at most the Arrwwid number with respect to squares. It follows that the Arrwwid number is the same regardless of whether disk or square ranges are considered.

The above definitions naturally extend to higher dimensions by replacing disks by balls, area by volume, and squares by (hyper)cubes.

\subsection{Related work}
Jagadish, Kumar and Moon et al.\ studied how well space-filling curves succeed in keeping the number of fragments needed to cover a query range low~\cite{Jagadish1990,Jagadish1997,Kumar1994,Moon2001}. The curve quality measures in their work are based on the number of fragments needed to cover a query range, averaged over a selection of query ranges that depends on the underlying tiling. As a result their results can only be used to analyse space-filling curves with the same underlying tiling; in particular they assume a tiling that subdivides squares into \emph{four} smaller squares, expanded down to a fixed level of recursion. This class of curves includes well-known curves such as the Hilbert curve~\cite{Hilbert1891} and Z-order, also known as Morton order or Lebesgue order~\cite{Lebesgue1904}. However, it does not include, for example, Peano's curve~\cite{Peano1890}, which is based on subdividing squares into \emph{nine} squares and seems to be the curve of choice in certain applications~\cite{Bader2006a,Voorhies1991}. The work by Jagadish, Kumar and Moon et al.\ does not enable a comparison between, for example, Hilbert's curve and Peano's curve.

The Arrwwid number, as defined above, does not have this limitation: it admits a comparison between curves with different underlying tilings. Nevertheless Asano et al., too, only studied curves based on the tiling with four squares per square~\cite{Asano1997}. The Arrwwid number of such a tiling is four. Asano et al.\ studied what can be achieved by controlling the order in which the tiles are stored: they presented an ordering scheme, the \ARRWW space-filling curve, that guarantees that whenever four tiles are needed, at least two of them are consecutive on disk. Thus these four tiles can be divided into at most three sets such that the tiles within each set are consecutive, and thus the Arrwwid number of the \ARRWW space-filling curve is three. Asano et al.\ also proved that one cannot do better: no ordering scheme of this particular tiling has Arrwwid number less than three.

In practice, the effectiveness of a curve in optimising disk access time will depend on the cost of ``jumping'' to another fragment relative to the cost of scanning a false answer. None of the curve quality measures in the papers cited above take this into account. Bugnion et al.~\cite{Bugnion1997} studied average disk access times as a function of the cost of jumping relative to scanning. Their work considers scanning orders that follow square grid tilings in such a way that each tile touches the previous tile in the scanning order. Note that all curves mentioned so far, except Z-order, are like that. In such tilings, consecutive tiles may be \emph{horizontally connected} (sharing a vertical edge), \emph{vertically} connected (sharing a horizontal edge) or they may be \emph{diagonally connected} (touching only in a vertex). Bugnion et al.\ analysed the performance of scanning orders under the assumption that they resemble random walks with a given proportion of horizontal, vertical and diagonal connections (and some further simplifying assumptions). Their results suggest that diagonal connections (as in the \ARRWW curve) are harmful for random walks, regardless of the cost of jumping relative to scanning.

\subsection{Results}
In this paper, we extend the scope of our knowledge on Arrwwid numbers to different tilings, which do not necessarily have four squares per square, and to higher dimensions. The results are the following: in two dimensions, no recursive tiling and no recursive space-filling curve has Arrwwid number better than three if the tiles are simply connected regions in the plane. This paper presents recursive tilings with Arrwwid number three (regardless of the ordering), as well as an alternative square-based space-filling curve with Arrwwid number three but without the diagonal connections that seem to hurt the performance of the \ARRWW curve.

In $d$ dimensions, no recursive tiling has Arrwwid number better than $d+1$. We prove that in three dimensions, putting the tiles in a certain order will not make it possible to get below this bound if the tiles are convex polyhedra. There are recursive tilings with fractal-shaped tiles that have Arrwwid number $d+1$, and tilings with rectangular blocks that have Arrwwid number $\frac34 \cdot 2^d$. In two dimensions, space-filling curves with optimal Arrwwid number can be constructed by defining a good order on a regular square-based tiling, but this does not generalise to three (or more) dimensions. Regular (Hyper)cube-based tilings have Arrwwid number $2^d$, and any space-filling curve based on such a tiling has Arrwwid number at least $2^d - 1$. For $d \geq 3$ this is more than what can be achieved with rectangular blocks.

\section{Two-dimensional tilings}

\subsection{Defining recursive tilings}
A recursive tiling of a region \unittile (called the \emph{unit tile}) in the plane is defined by a finite set of recursive tiling rules. Each rule specifies (i) the shape of the region to be tiled; (ii) how this region is tiled with a fixed number of tiles; (iii) which rules should be applied to subdivide each of these tiles recursively. Figure~\ref{fig:tilings} shows a number of examples. Each rule is identified by a letter, and depicted by drawing the shape of its region, the tiles, and within the tiles, the letters of the rules to be applied to them; each letter is rotated and/or mirrored to reflect the transformations that should be applied to the subtiles of the tile. Next to each set of rules we see the tiling that is produced after expanding the recursion down to tiles of a few millimeters.

\begin{figure}
\centering
\includegraphics[width=\hsize]{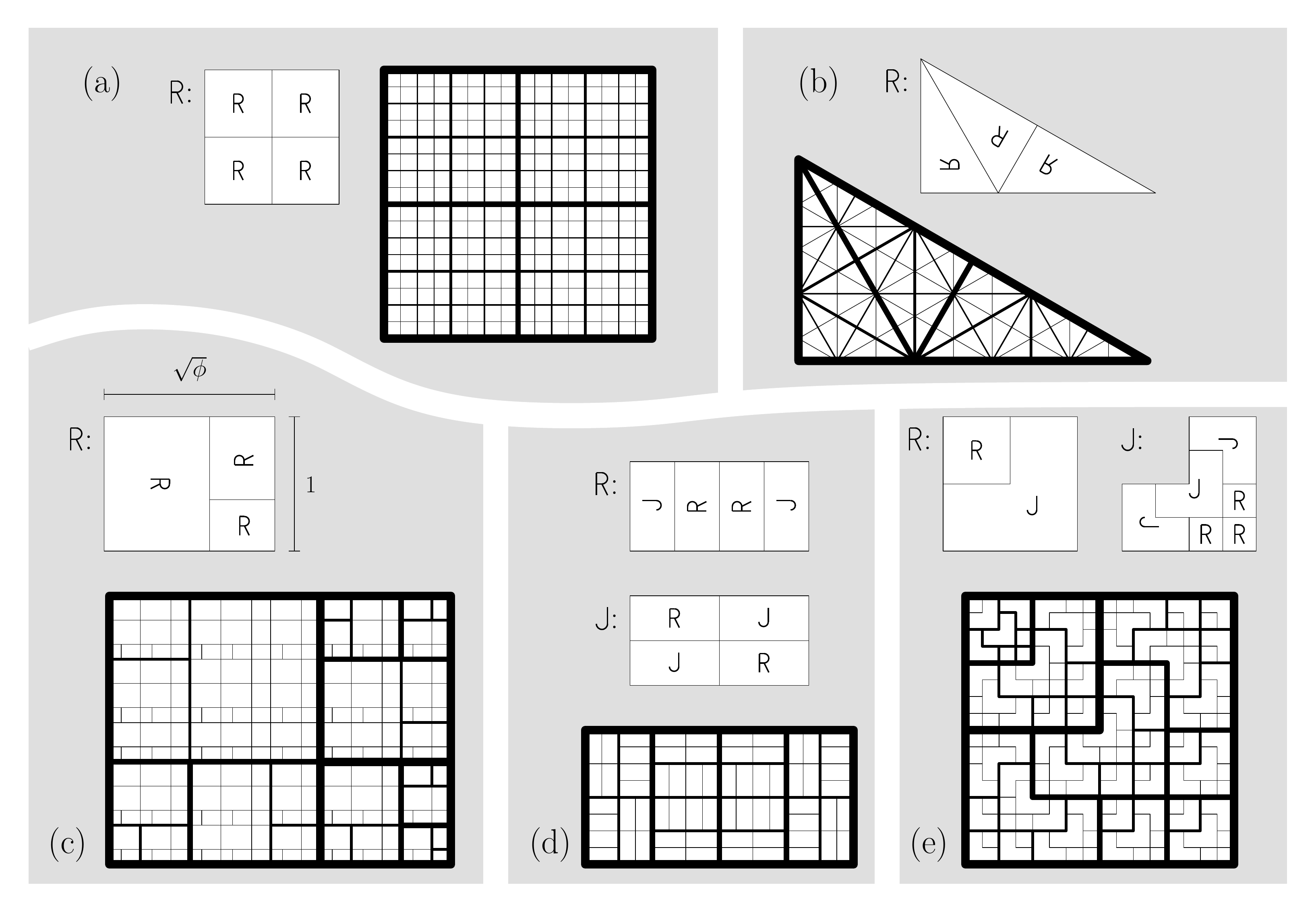}
\caption{Examples of recursive tilings. (a) A simple, uniform, regular recursive tiling. (b) A simple, uniform recursive tiling. (c) A simple, non-uniform recursive tiling based on the golden ratio $\phi = \frac12 \sqrt{5} + \frac12$. (d) Composite and uniform. (e) Composite, non-uniform.}
\label{fig:tilings}
\end{figure}

\emph{Simple} tilings are recursive tilings that only use one rule (Figure~\ref{fig:tilings}(a,b,c)). Using a term coined by Mandelbrot~\cite{Mandelbrot1983}, we could also call them \emph{pertilings}. In contrast, \emph{composite} tilings are recursive tilings that use multiple rules (Figure~\ref{fig:tilings}(d,e)).
We define \emph{uniform} tilings as recursive tilings in which all tiles have the same shape, and each rule subdivides such a shape into an equal number of tiles of equal size (Figure~\ref{fig:tilings}(a,b,d)).
By \emph{square}, \emph{rectangular}, \emph{triangular} and \emph{fractal} tilings we mean tilings whose tiles are square, rectangular, triangular, or fractal-shaped.
\emph{Regular} recursive tilings are recursive tilings that form a fully regular grid or tiling when the recursion is expanded to any fixed depth (Figure~\ref{fig:tilings}(a)).

The \emph{size} of a uniform tiling is the number of tiles in each rule.
For any given (recursive or non-recursive) tiling, the \emph{degree} of a point $p$ in \unittile is the maximum number of interior-disjoint tiles that meet in $p$.
%; in other words, it is the maximum size of any set of tiles that meet in $p$ and such that no tile is an ancestor of another in the recursion.
The \emph{vertex degree} of a (recursive or non-recursive) tiling is the maximum degree of $p$ over all points $p \in \unittile$.

The results in this paper apply to simple as well as composite tilings, and uniform as well as non-uniform tilings. We only require that the number of recursion rules is finite, that each rule that defines the tiling starts with a finite region, and that the transformations that may be applied to the rules preserve similarity.

\subsection{The Arrwwid number of a tiling: basic bounds}
\begin{observation}\label{ob:squaretiling}
The uniform square tiling of Figure~\ref{fig:tilings}(a) has Arrwwid number four.
\end{observation}
\begin{proof}
We first prove that the Arrwwid number of the uniform square tiling is at most four.
Let $r$ be the radius of any disk $Q$. Since the tiles differ in width by a factor two from one level of recursion to the next, there is a level of recursion where the tiles have side lengths between $2r$ and $4r$, and hence, area between $4r^2$ and $16r^2$. The distance between any pair of horizontal lines of the grid formed by these tiles is at least $2r$, so $Q$ is crossed by at most one horizontal grid line. By the same argument, $Q$ is crossed by at most one vertical grid line. It follows that at most four tiles of this grid intersect $Q$, so $Q$ can always be covered by at most four tiles with total area at most $4 \cdot 16r^2 = 64r^2$. Hence the Arrwwid number of the square tiling is at most four (with cover ratio at most $64r^2 / \pi r^2 = 64/\pi$).

It remains to show that the Arrwwid number of the uniform square tiling cannot be three or less. Suppose the Arrwwid number would be at most three, so there is a constant $c$ such that any disk of radius $r$ can be covered by at most three tiles with total area at most $cr^2$. Consider a fixed value $r$, a tile $T$ with area more than $c r^2$, and a disk $Q$ of radius $r$ placed in the middle of this tile~$T$, that is, centered on the spot where the four subtiles of $T$ meet. Thus $Q$ intersects each of the four subtiles of $T$, and the only way to cover it with less than four tiles is by covering it with $T$ itself---but the area of $T$ is larger than $c r^2$. This contradicts our assumptions and thus proves that the Arrwwid number of the uniform square tiling is at least four.
\end{proof}

The above arguments illustrate that there is a relation between the vertex degree of a tiling and its Arrwwid number:
\begin{observation}
The Arrwwid number of a recursive tiling cannot be lower than its vertex degree.
\end{observation}
This is because any sufficiently small circle centered on a point of maximum degree would yield a contradiction to the assumption that the Arrwwid number is lower.

The subdivision of the unit tile \unittile into its subtiles creates at least one curve $c$ in the interior of~\unittile. Since tile sizes decrease by at least a constant factor with each level of recursion, further down in the recursion $c$ must eventually be subdivided by crossings or T-junctions with curves that divide the tiles on each side of $c$. Thus we get points on $c$ of degree at least three. Hence we get the following:

\begin{theorem}\label{th:2dtilinglb}
Each recursive tiling of a two-dimensional region \unittile has Arrwwid number at least three.
\end{theorem}

We now ask: are there actually any recursive tilings with Arrwwid number three? In the next subsections, we answer this question with ``yes''.

\subsection{Recursifying non-recursive tilings}
As observed above, the Arrwwid number of a tiling is at least the vertex degree of the tiling. Therefore, to find a tiling with Arrwwid number three, we need to find a tiling with vertex degree at most three. A regular hexagonal tiling has this property, but that tiling is not recursive: a hexagon cannot be subdivided into similar hexagons. We will now see how a recursive tiling can be obtained by recursively approximating a hexagon by smaller hexagons; this procedure will turn the boundaries of the hexagons into fractals. The technique itself is not new; however, I am currently not aware of any published general description such as the one given below.

The procedure is as follows. We start with a \emph{coarse tiling} which is a regular tiling of large hexagonal tiles. We overlay this with a \emph{fine tiling} of small hexagonal tiles, such that the fine tiling looks the same around each large tile (Figure~\ref{fig:recursification}(a)). Now we assign each small tile to a large tile, in such a way that the union of the small tiles assigned to each large tile approximates the shape of the large tile well. To accomplish this we assign each small tile that is completely contained in a single large tile $T$ to $T$. The remaining small tiles are assigned according to a tie-breaking rule that ensures that for each large tile the union of the small tiles assigned to it has the same size and shape. In Figure~\ref{fig:recursification}(b) we do this as follows: when going clockwise around a large tile $T$, starting at three o'clock, we alternate between giving a small tile to a neighbour of $T$ and assigning a small tile to $T$. All large tiles of the coarse grid are now replaced by the union of the small tiles assigned to them (Figure~\ref{fig:recursification}(c)), thus adding detail to the boundaries of the large tiles. We now replace the small tiles by scaled copies of the large tiles (which are no longer hexagons), and again replace the large tiles by the unions of the small tiles assigned to them. This adds even more detail to the boundaries of the large tiles (Figure~\ref{fig:recursification}(d)). When we repeat this process ad infinitum, the boundaries of the tiles will converge to fractal shapes such that each large tile is tiled by nine scaled-down copies of itself (Figure~\ref{fig:recursification}(e)).

\begin{figure}
\centering
\hbox to \hsize{\hss\includegraphics[width=1.1\hsize]{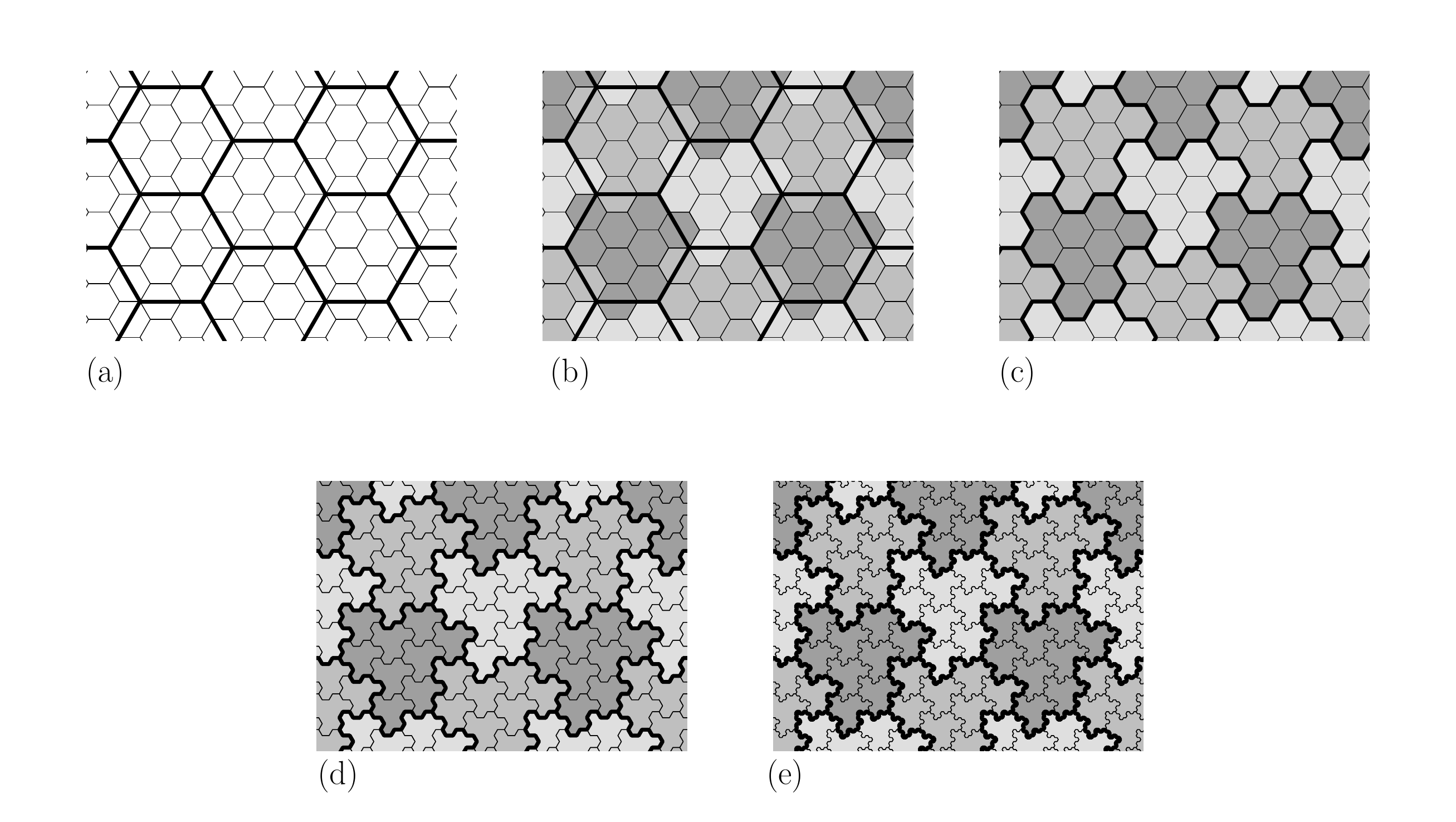}\hss}
\caption{How to turn a hexagonal tiling into a recursive tiling of size nine with Arrwwid number three.}
\label{fig:recursification}
\end{figure}

The transformation described above changed the shape of the tiles in the coarse tiling, but it did not change its topological structure. The vertex degree of the resulting tiling is still three, but unlike the original hexagonal tiling, our fractal tiling is recursive. In fact, we have produced a fractal tiling with Arrwwid number three (we will see how to prove such things below).

\begin{figure}
\centering
\hbox to \hsize{\hss\includegraphics[width=1.1\hsize]{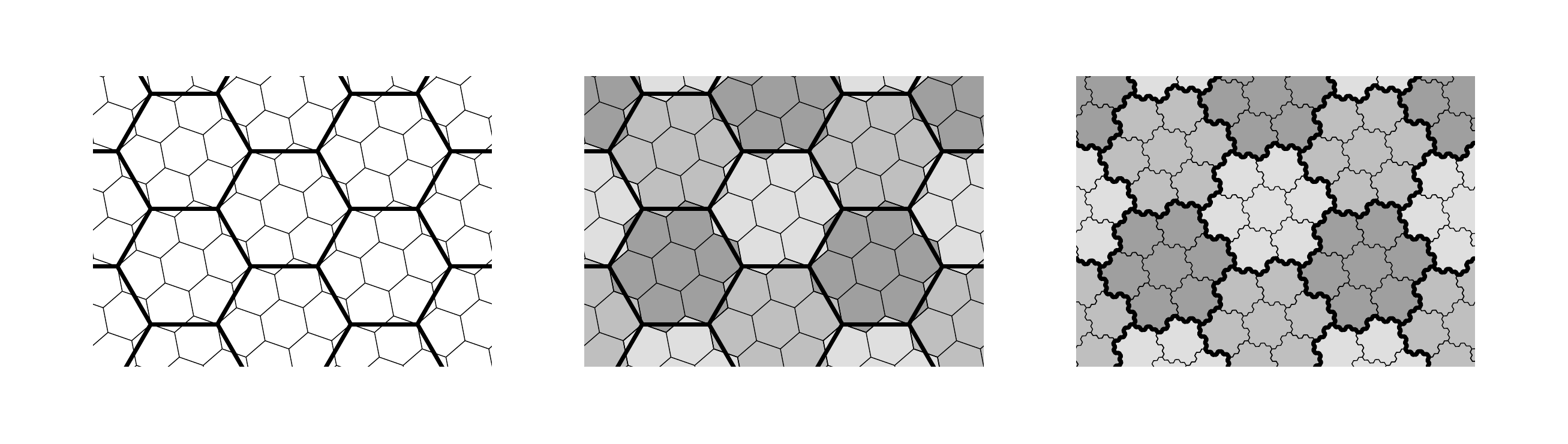}\hss}
\caption{A recursive tiling of Gosper islands.}
\label{fig:gosper}
\end{figure}

\begin{figure}
\centering
\hbox to \hsize{\hss\includegraphics[width=1.1\hsize]{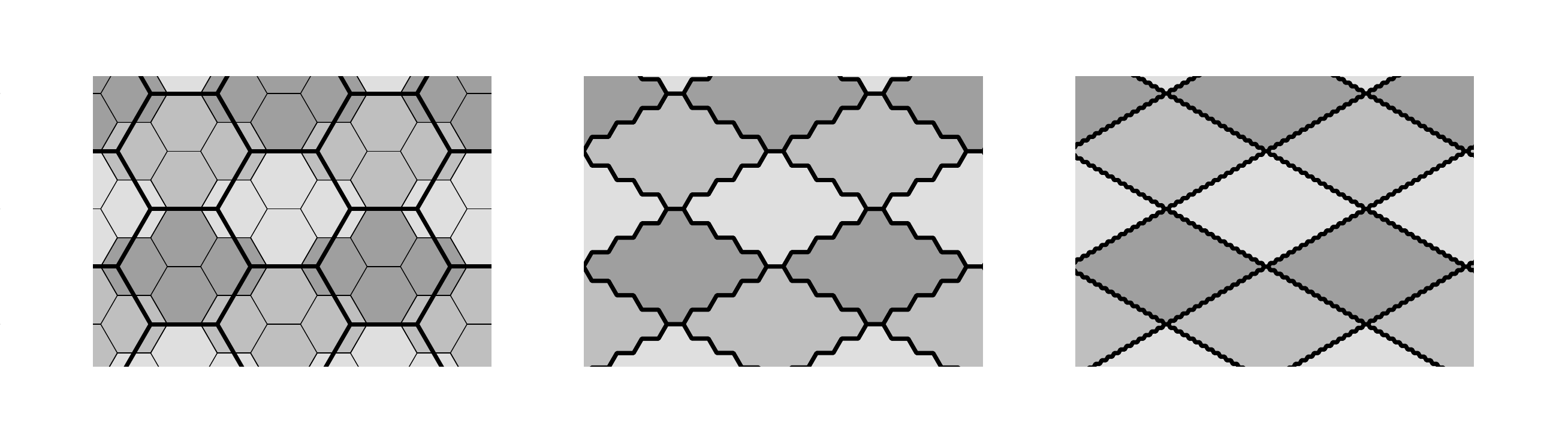}\hss}
\caption{A failed attempt to turn a hexagonal tiling into a recursive tiling of size four with Arrwwid number three: the Arrwwid number turns out to be four.}
\label{fig:rhombus}
\end{figure}

\begin{figure}
\centering
\hbox to \hsize{\hss\includegraphics[width=1.1\hsize]{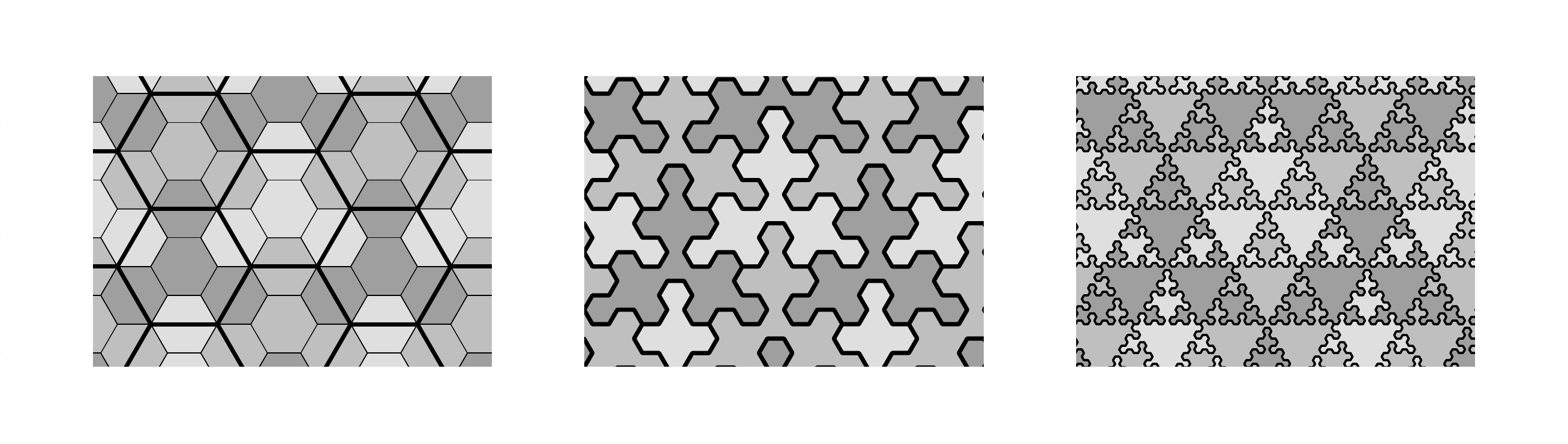}\hss}
\caption{A failed attempt to turn a hexagonal tiling into a recursive tiling of size four with Arrwwid number three: the Arrwwid number turns out to be four.}
\label{fig:disconnected}
\end{figure}

The technique described above requires some care: it matters how well the small hexagonal tiles approximate the large hexagonal tiles. Figure~\ref{fig:gosper} shows a recursification with seven small tiles per large tile, resulting in a tiling with shapes called Gosper islands~\cite{Gardner1976}. Figure~\ref{fig:rhombus} shows an attempt with four small tiles per large tile. However, it fails to produce a fractal tiling with Arrwwid number three: the approximation of large tiles by small tiles is so bad that the hexagonal tiling morphs into a regular grid of rhombuses with vertex degree four. Figure~\ref{fig:disconnected} shows another failed attempt: here the result is a tiling whose tiles have disconnected interiors, and it contains vertices where three different tiles meet the three largest connected components of a fourth tile.

\begin{theorem}\label{th:gosper}
A recursive tiling with Gosper islands has Arrwwid number three.
\end{theorem}
\begin{proof}
Assume that in the course tiling neighbouring vertices are at a distance $u$ from each other. In the first stage of refinement, each segment $e$ of a large tile boundary is replaced by three segments of length~$u/\sqrt{7}$, the first and last of which make an angle of $\arctan{\frac15 \sqrt{3}}$ with $e$. The new boundary of the large tile thus stays within a distance $d_1 = \sin(\arctan{\frac15 \sqrt{3}})\cdot u/\sqrt{7} = \frac{u}{14}\sqrt{3}$ of the old boundary. Refining the tiles $i$ times recursively thus keeps the boundary within a distance $d_i$ that satisfies $d_i \leq \frac{u}{14}\sqrt{3} + d_{i-1}/\sqrt{7}$; for $i \rightarrow \infty$ this converges to $d_i < 0.199 \cdot u$. Let $s$ and $s'$ be the radius of the smallest disk that intersects more than three large tiles in the original and in the recursified large tiling, respectively. We have $s = u/2$, and since the tile boundaries in the recursified tiling are within a distance of $0.199 \cdot u$ from the tile boundaries in the original tiling, we have $s' > s - 0.199 \cdot u = 0.301 \cdot u$.

Now consider any disk $Q$ of radius $r$. There is a level of recursion in the recursified tiling where $u$, the distance between neighbouring vertices in the underlying hexagonal tiling, is between $r/0.301$ and $\sqrt{7}\cdot r/0.301$. Since $r \leq 0.301 u < s'$, the disk $Q$ intersects at most three of the tiles on this level of the recursified tiling, and each of these tiles has area $\frac32 \sqrt{3} \cdot u^2  \leq \frac32 \sqrt{3} \cdot 7 r^2 / 0.301^2 < 201 \cdot r^2$ (equal to the hexagon whose position they are taking). It follows that $Q$ can always be covered by three tiles of total area at most $603 \cdot r^2 < 192 \cdot \area(Q)$; hence the Arrwwid number of the Gosper tiling is three.
\end{proof}

\begin{figure}
\centering
\hbox to \hsize{\hss\includegraphics[width=1.1\hsize]{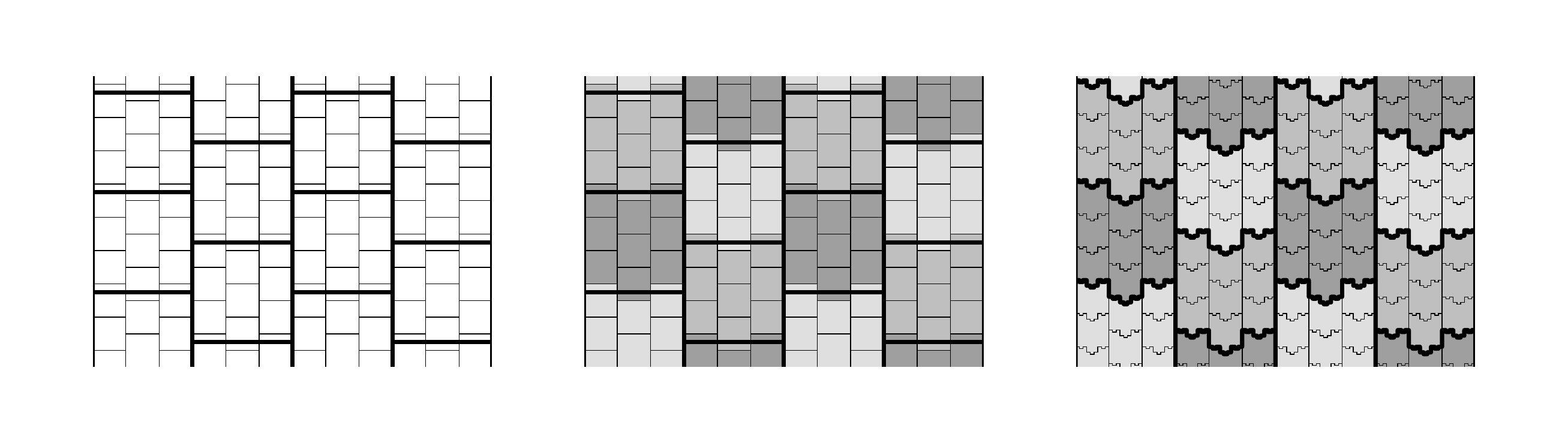}\hss}
\caption{A recursive tiling with Arrwwid number three based on a ``shifted'' grid of squares.}
\label{fig:shiftedsquares}
\end{figure}

\enlargethispage{\baselineskip}
In a similar way one could prove that the tiling of Figure~\ref{fig:recursification}(e) has Arrwwid number three. The recursification technique can also be applied to non-hexagonal tilings. A recursified grid of squares with five small tiles per large tile results in a tiling of cross-like shapes~(see \cite{Mandelbrot1983} plate 49, or \cite{Teachout2009}). Other examples with hexagons can be seen to underly the ``generalised Gosper curves'' of Akiyama et al.~\cite{Akiyama2005}. Figure~\ref{fig:shiftedsquares} shows a recursified tiling based on a grid of squares of which each column is shifted by half a square's height with respect to the next column; it has Arrwwid number three.

\subsection{An optimal rectangular tiling}

The fractal tilings are beautiful, but have one big practical disadvantage: for any given point it is hard to figure out in which tiles it lies. It would be quite challenging to sort any given set of data points into fractal tiles, so it would be hard to use such tilings to build a data structure. Therefore it may be interesting to look for a tiling with Arrwwid number three that uses simpler tiles, namely rectangles.

In a uniform rectangular tiling, each rule defines a subdivision of a large rectangle into, say, $t$ similar, smaller rectangles which are scaled by a factor $1/\sqrt{t}$ with respect to the large rectangle. Let $\alpha$ be the width/height ratio of the large rectangle; without loss of generality, assume $\alpha \geq 1$.

For any horizontal line through the large rectangle, the smaller rectangles along that line must fill the full width of the large rectangle. This implies that for fixed $\alpha$ and $t$,  the equation\begin{equation}\label{eq:fillwidth}
\alpha n_{ww} + n_{hw} = \alpha \sqrt{t}
\end{equation}
has at least one non-negative integer solution for the variables $n_{ww}$ and $n_{hw}$. Similarly, considering vertical lines we find that\begin{equation}\label{eq:fillheight}
\alpha n_{wh} + n_{hh} = \sqrt{t}
\end{equation}
must have a non-negative integer solution.
In fact, to get an Arrwwid number of three, for each of the variables $n_{ww}$, $n_{hw}$, $n_{wh}$ and $n_{hh}$ there must be a solution to the above equations where this variable is at least one. Otherwise the only way to subdivide the large rectangle would be to put the tiles in a regular rectangular grid, which would result in an Arrwwid number of four.

\begin{lemma}
The size $t$ of any uniform rectangular tiling with Arrwwid number three is a square number.
\end{lemma}
\begin{proof}
For the sake of contradiction, assume that $t$ is not a square number. Then any solution to Equation~\ref{eq:fillheight} must have $n_{wh} > 0$; this holds in particular for any solution with $n_{hh} > 0$. Fix $n_{wh} > 0$ and $n_{hh} > 0$ such that they satisfy Equation~\ref{eq:fillheight}. We get:\begin{equation}\label{eq:alpha}
\alpha = \frac{\sqrt{t} - n_{hh}}{n_{wh}}.
\end{equation}
Substituting $\alpha$ in Equation~\ref{eq:fillwidth} yields:\begin{equation}
\frac{n_{ww}}{n_{wh}}\sqrt{t} - \frac{n_{ww}n_{hh}}{n_{wh}} + n_{hw} =
\frac{t}{n_{wh}} - \frac{n_{hh}}{n_{wh}}\sqrt{t}.
\end{equation}
Because $\sqrt{t}$ is irrational, the terms that include $\sqrt{t}$ must be equal, that is, $n_{ww} = -n_{hh}$, but then $n_{ww}$ is negative: a contradiction. Hence $t$ must be a square number.
\end{proof}

\begin{lemma}
The width/height ratio $\alpha$ of the rectangles in a uniform rectangular tiling of size $t$ with Arrwwid number three is a rational number with numerator at most $\sqrt{t}$ and denominator less than $\sqrt{t}$.
\end{lemma}
\begin{proof}
We have $\alpha > 1$, otherwise we would get a regular grid of squares with Arrwwid number four. Therefore the solution to Equation~\ref{eq:fillwidth} for which $n_{wh} > 0$, must have $n_{wh} < \sqrt{t}$ and $\alpha = (\sqrt{t} - n_{hh})/n_{wh}$. %Hence $\alpha$ is a rational number with numerator at most $\sqrt{t}$ and denominator less than $\sqrt{t}$.
\end{proof}

The above two lemmas brought an exhaustive search for increasing values of $t$ within reach, trying all eligible values of $\alpha$ for each $t$. This led to the following result:

\begin{figure}
\centering
\includegraphics[height=4cm]{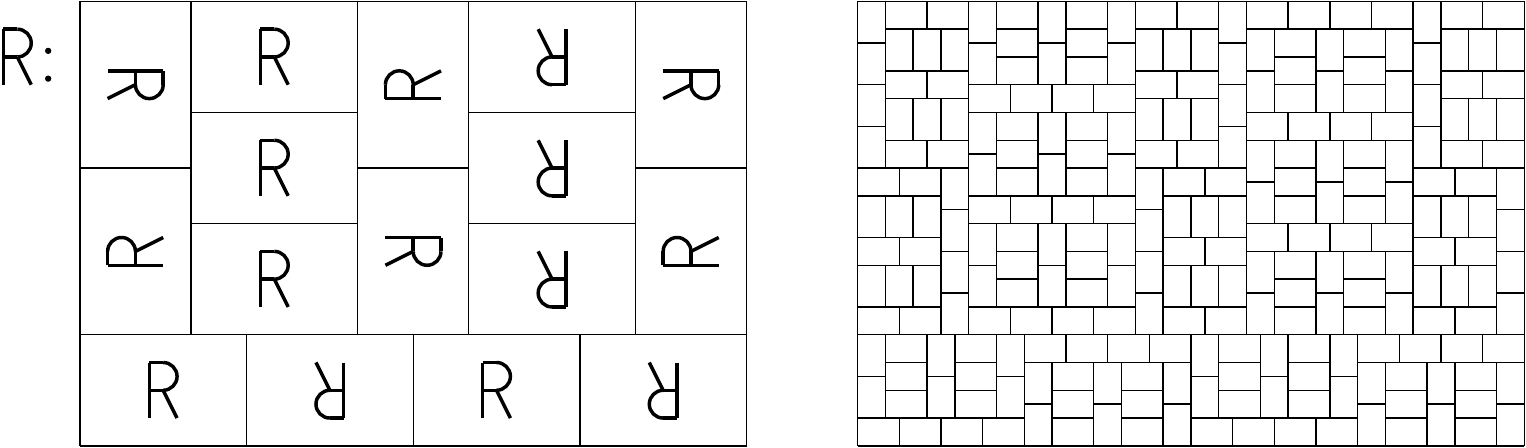}
\caption{Definition of the Daun tiling, and level-two expansion.}
\label{fig:daun}
\end{figure}

\begin{theorem}\label{th:Dauntiling}
The smallest uniform rectangular tiling (with fewest tiles in the definining rules) with Arrwwid number three is the Daun tiling, shown in Figure~\ref{fig:daun}.
\end{theorem}
\begin{proof}
We need to prove (i) the Daun tiling has Arrwwid number three, and (ii) no smaller uniform rectangular tiling has Arrwwid number three. The second part is tedious and boring but doable by hand; skeptic readers are welcome try by themselves or to request a photocopy of my notes. We will now prove the first part: the Daun tiling has Arrwwid number three.

Consider a disk $Q$ of radius $r$, and let $k_r$ be the level of recursion on which the tiles have dimensions more than $6r \times 4r$ and at most $24r \times 16r$. We define a level-$k$ vertex, edge or tile as a vertex (a meeting point of three or more tiles), edge or tile that first appears in the tiling on the $k$-th level of subdivision down from the unit rectangle. So the corners of the unit tile (the level-0 tile) are level-0 vertices, the remaining corners of its sixteen subtiles are level-1 vertices, and so forth. Consider the boundary of any tile to be labelled as follows (see Figure~\ref{fig:daun-labelling}(a)): starting from the marked corner (the one that was the lower left corner before rotation) and going in clockwise direction, we divide the boundary of the tile into ten sections of equal length, which we label B, A, B, C, D, E, D, E, F, and A, in that order. We will now prove by induction that there are no points of degree four in the tiling ${\cal T}$ that results from expanding the recursion to a depth of $k_r$ levels.
\begin{figure}
\centering
\hbox to\hsize{\hss\includegraphics[width=\hsize]{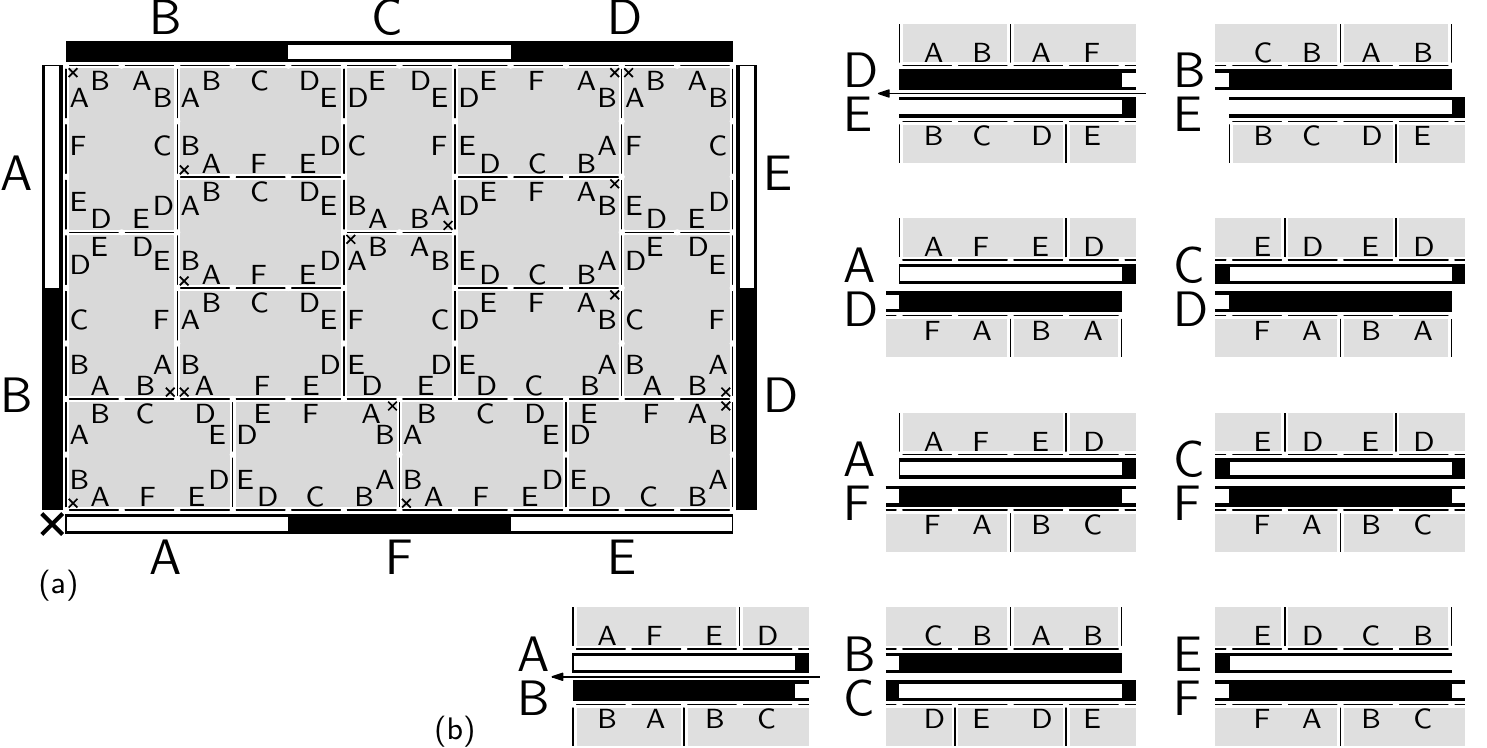}\hss}
\caption{(a) Labelling the boundary of the Daun tiling. (b) Matching sections to each other.}
\label{fig:daun-labelling}
\end{figure}
Let $H(k)$ be the following claim:\begin{itemize}
\item[(i)] there is no point of degree four at the starting point or along the interior of an A-section of a level-$k$ tile that coincides with a B-section of another level-$k$ tile;
\item[(ii)] there is no point of degree four at the starting point or along the interior of a D-section of a level-$k$ tile that coincides with an E-section of another level-$k$ tile;
\item[(iii)] for all $(x,y) \in \{
    (\textrm{A},\textrm{D}),
    (\textrm{A},\textrm{F}),
    (\textrm{B},\textrm{C}),
    (\textrm{B},\textrm{E}),
    (\textrm{C},\textrm{D}),
    (\textrm{C},\textrm{F}),
    (\textrm{E},\textrm{F})
    \}$, there is no point of degree four along or at any endpoint of an $x$-section of a level-$k$ tile that coincindes with a $y$-section of another level-$k$ tile;
\item[(iv)] subdividing a level-$k$ tile does not create any points of degree four in its interior.
\end{itemize}
We first prove $H(k_r)$. Part (iv) is trivial because level-$k_r$ tiles are not subdivided further in ${\cal T}$. Part (i) is true because neither the A-section nor the B-section is subdivided further, and at the end of any B-section (which coincides with the starting point of the other tile's A-section) the tile boundary of which it is a part continues in a straight line, so that there cannot be a vertex of degree four there. Part (ii) and (iii) can be verified in a similar way.

We will now prove the following for $k < k_r$: if $H(k+1)$ is true, then $H(k)$ is true. Part (i): As illustrated in Figure~\ref{fig:daun-labelling}(a), an A-section $s$ at level $k$ is subdivided into a D-section, an E-section, an F-section and an A-section (going clockwise around the tile) at level $k+1$. A B-section at level $k$ consists of a C-section, a B-section, an A-section, and a B-section (going counterclockwise around the tile). Thus, along an A-section of a level-$k$ tile that coincides with a B-section of another level-$k$ tile, a D-section is matched to a C-section at level $k+1$, E to B, F to A, and A to B (Figure~\ref{fig:daun-labelling}(b)). Because $H(k+1)$ is true, this does not create a degree-four vertex at the start or anywhere along the interior of $s$. Part (ii) and (iii) can be verified in a similar way. For part (iv), it follows from $H(k+1)$, part (iv), that no points of degree four are created in the interior of level-$(k+1)$ tiles, but we have to be careful about points on the boundary of level-$(k+1)$ tiles. Figure~\ref{fig:daun-labelling}(a) illustrates that the boundary sections of level-$(k+1)$ tiles inside a level-$k$ tile only coincide in the combinations listed in $H(k+1)$, part (i), (ii) and (iii); therefore the interiors of these level-$(k+1)$ boundary sections do not contain any points of degree four. All end points of these boundary sections that lie in the interior of a level-$k$ tile are a starting point of a level-$(k+1)$ A-section coinciding with a B-section, a starting point of a level-$(k+1)$ D-section coinciding with an E-section, or any end point of a combination of level-$(k+1)$ sections as listed in part (iii) of $H(k+1)$. Therefore, by $H(k+1)$, these points cannot have degree four either. This proves part (iv) of $H(k)$.

It follows by induction that $H(0)$ holds, that is, ${\cal T}$ does not have any degree-four vertices in the interior of the unit rectangle. Obviously there are not any degree-four vertices on its boundary either. Now observe that all tiles of ${\cal T}$ consist of six squares of a regular square grid with line spacing more than $2r$. Any disk $Q$ of radius $r$ intersects at most four of these squares. If $Q$ does indeed intersect as much as four squares, then the number of tiles of ${\cal T}$ that intersect $Q$ is equal to the degree of the grid point $p$ shared by these four squares. As we have just established, the degree of any point $p$ in ${\cal T}$ is at most three. Hence $Q$ can be covered by at most three tiles of ${\cal T}$, which have total area at most $3 \cdot 24r \cdot 16r = 1152r^2 = \frac{1152}\pi \area(Q)$. This proves that the Arrwwid number of the Daun tiling is three.
\end{proof}

\section{Two-dimensional space-filling curves}
\label{sec:2DSFC}

\begin{figure}
\centering
\includegraphics[height=5cm]{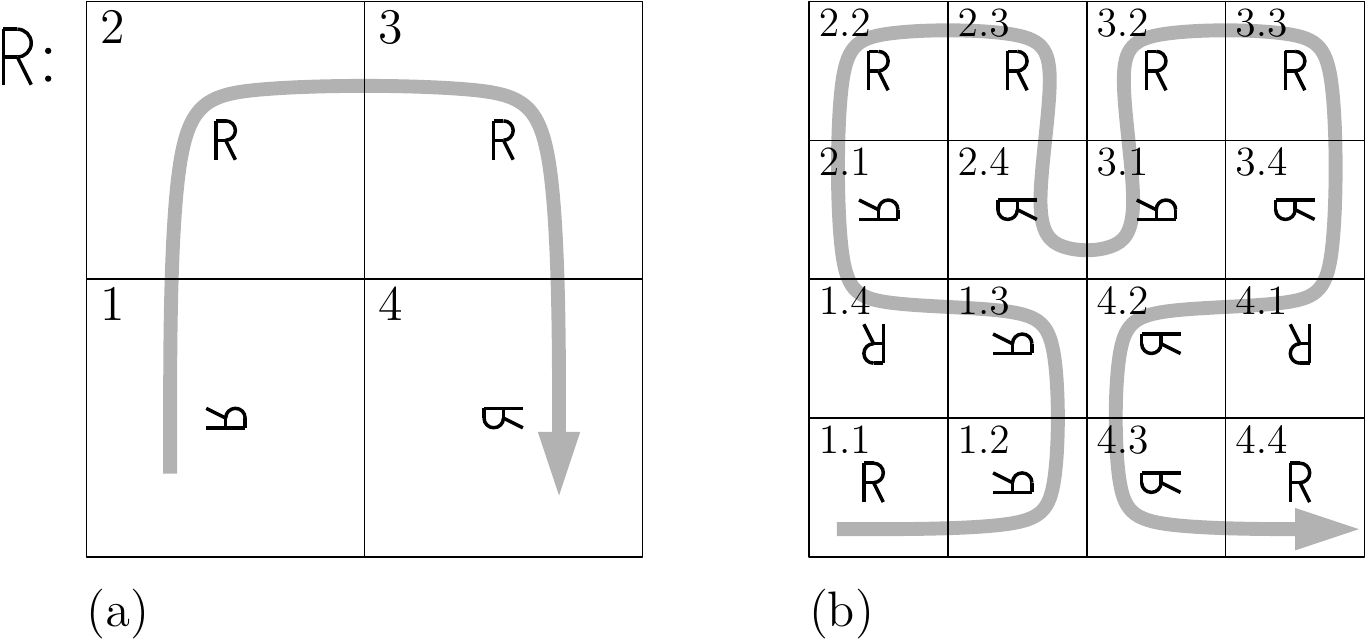}
\caption{(a) Definition of the Hilbert scanning order. (b) Definition expanded by one level.}
\label{fig:hilbert}
\end{figure}

\subsection{Defining space-filling curves}
\label{sec:2DSFC-definition}
We define a \emph{recursive scanning order} as a recursive tiling in which the rules specify not only a subdivision into tiles and what rules to apply to the tiles, but also an order of the tiles. One may illustrate this by a curve that visits the tiles recursively in order. Figure~\ref{fig:hilbert}(a) shows an example: the Hilbert curve~\cite{Hilbert1891}. When the tiling is refined by expanding the recursion, so is the curve, see Figure~\ref{fig:hilbert}(b). When defining a scanning order, the recursion within a tile can be rotated, mirrored and scaled as with recursive tilings. In addition it is possible to apply a recursive rule with reversed order; we indicate a reverse application of a recursive rule by a horizontal stroke above the letter that identifies the rule (see Figure~\ref{fig:dekking} for an example).

One may use a scanning order to define a mapping from the unit interval to the unit tile as follows. For any tile $A$ and any $x \in [0,1]$, we define the \emph{prefix region} $\prefix(A,x)$ as the region within $A$ that has total area $x \cdot \area(A)$ and comes first in the scanning order; intuitively, this region can be constructed by subdividing $A$ recursively to a sufficiently fine level and collecting tiles in scanning order, until tiles with a total area of $x \cdot \area(A)$ have been collected. We can define this more precisely as follows. If $x = 0$, then $\prefix(A,x) = \emptyset$. Otherwise let $A_1,...,A_k$ be the subtiles of $A$ in order, let $a_1,...,a_k$ be their areas relative to $A$, that is, $a_i = \area(A_i)/\area(A)$, and let $c_0,...,c_k$ be their cumulative areas relative to $A$, that is, $c_i = \sum_{j=1}^i a_j$. Now let $i$ be the largest $i$ such that $c_{i-1} \leq x$. Then $\prefix(A,x) = A_1 \cup ... \cup A_{i-1} \cup \prefix(A_i,(x-c_{i-1})/a_i)$. Let the \emph{postfix region} $\postfix(A,x)$ be defined as $A \setminus \prefix(A,x)$, and the \emph{fragment} $A[x,y]$ as $\postfix(A,x) \cap \prefix(A,y)$.

Let \unittile be the unit tile. For fixed $y$, the fragment $\unittile[x,y]$ shrinks to a point as $x$ approaches $y$ from below; denote this point by $\sigma_{\uparrow}(y)$. Similarly, let $\sigma^{\downarrow}(x)$ be the point to which $\unittile[x,y]$ shrinks when $y$ approaches $x$ from above; see Figure~\ref{fig:bridges} for an example. Together the functions $\sigma_{\uparrow}$ and $\sigma^{\downarrow}$ define a ``curve'' with a discontinuity (a jump) for every $x$ such that $\sigma_{\uparrow}(x) \neq \sigma^{\downarrow}(x)$. These functions also constitute a surjective map from the unit interval to the full set of points within \unittile. Thus they ``fill'' \unittile: it is a \emph{space-filling curve}.\footnote{Of course real curves do not jump. When $\sigma_{\uparrow}(x) = \sigma^{\downarrow}(x)$ for all $x$, either function defines a proper space-filling curve. However, when there are jumps, our definitions only provide surjective maps from the unit interval to the points within \unittile, but they are not continuous, so they are not curves, and they do not match the definition of space-filling curves as one would find them in the mathematics literature. Z-order is an example of a scanning order with many jumps. Lebesgue shows how to get a proper (continuous) curve of the Z-order. In effect he defines $\prefix(A,x)$ in another way, to include connecting pieces of curve between the tiles~\cite{Lebesgue1904}. The same technique could be applied to make other curves with jumps continuous. However, for our applications it is more convenient to work with a representation without the connecting pieces.}

In this paper we identify a scanning order with the space-filling curve defined by it; we use the terms \emph{(scanning) order} and \emph{(space-filling) curve} interchangeably. The terminology introduced to describe tilings will also apply to curves based on those tilings. Simple curves use one rule of recursion; composite curves use multiple rules. Uniform or square curves are curves based on uniform or square tilings, respectively.

For any fragment $\unittile[x,y]$, we say that $\sigma^{\downarrow}(x)$ is the \emph{entry point} of the fragment, while $\sigma_{\uparrow}(y)$ is its \emph{exit point}. We say that two consecutive fragments $\unittile[x,y]$ and $\unittile[y,z]$ \emph{connect} in a point~$p$ if $p = \sigma_{\uparrow}(y) = \sigma^{\downarrow}(y)$. Since fragments of space-filling curves fill two-dimensional regions, we measure the size of fragments by area, not by length.

There are other approaches to describing space-filling curves than the one chosen above. For example, Peano, who was the first to invent a space-filling curve, described his curve in an algebraic way~\cite{Peano1890}. Other authors describe a curve by defining polygonal approximation of it, with a rule on how to refine each segment of the approximating polyline recursively~\cite{Gardner1976}. Many authors specify the regions filled by fragments of the curve (like we do) together with the location of the entry and exit points of such fragments, but without making reversals explicit (for example \cite{Asano1997,Sagan1994,Wierum2002}). Since we are concerned with the use of space-filling curves as a way to order points in the plane, we choose a method of description that explicitly defines how to order the space inside a unit tile~\cite{Haverkort2009}.
In the following, when definitions are given for curves from other authors, these definitions are the result of analysing the original descriptions to find a definition of a corresponding scanning order in our notation.

\subsection{The Arrwwid number of a space-filling curve}
Consider a space-filling curve that fills a region \unittile. Recall that the Arrwwid number of the curve is the smallest number $a$ such that there is a constant $c$ such that any disk $Q$ that lies entirely in \unittile can be covered by $a$ curve fragments with total area at most $c \cdot \area(Q)$. Since every tile in the recursive tiling underlying a space-filling curve is by itself a fragment of the curve, the Arrwwid number of a curve is never more than the Arrwwid number of the underlying tiling. However, the Arrwwid number of a curve may be \emph{less} than the Arrwwid number $a$ of the underlying tiling. This can happen if, whenever $a$ tiles are needed to cover a disk, some of these tiles are consecutive along the curve. To illustrate this, a space-filling curve that has a smaller Arrwwid number than the underlying tiling is described in the next subsection.

\subsection{Square curves with Arrwwid number three}
\label{sec:squarewithA3}

\begin{figure}
\centering
\includegraphics[width=\hsize]{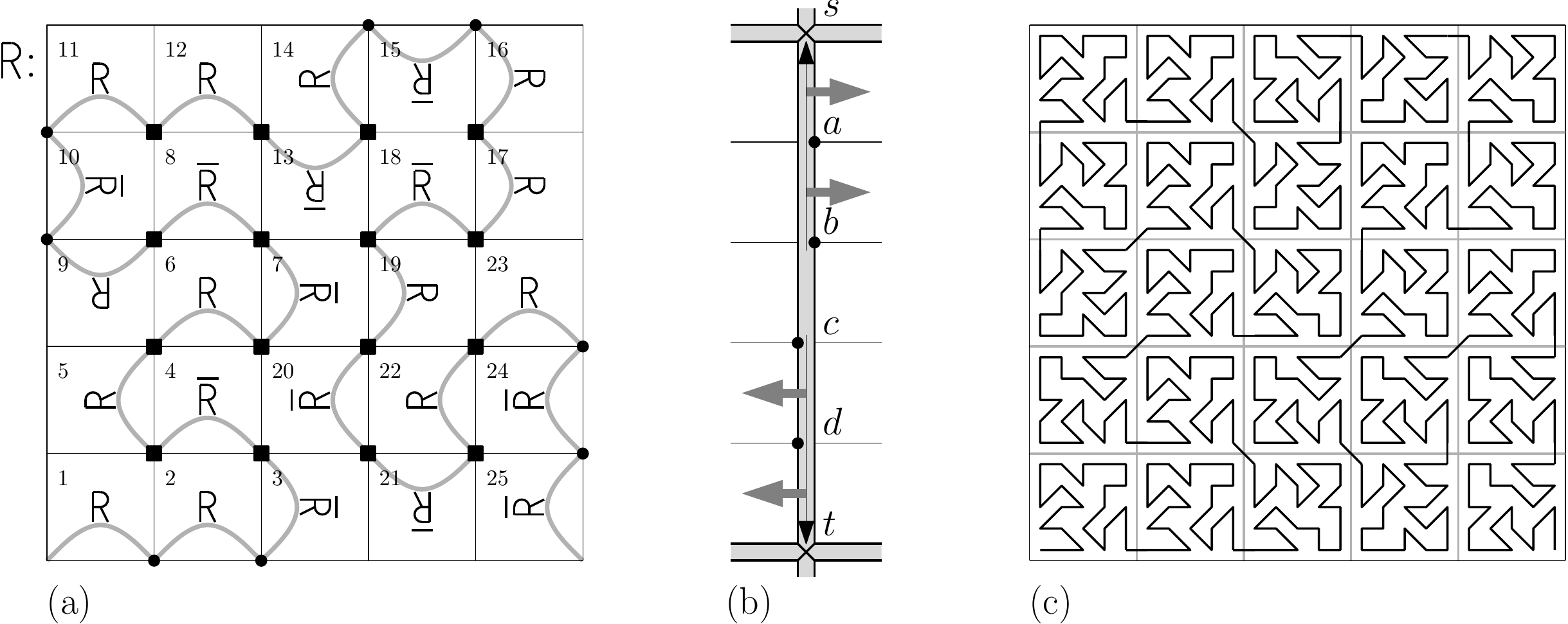}
\caption{(a) Definition of a scanning order following Dekking's curve~\cite{Dekking1982}, with marks added to aid in the proof of Theorem~\ref{th:dekking}. (b) When $p$ lies between two tiles ($p \in \{a,b,c,d\}$), we walk to the closest higher-level vertex $q \in \{s,t\}$ and consider $p$ as one of the vertices of the tile that lies to the right. (c) Sketch of the level-two expansion of the scanning order. This curve looks quite different from Dekking's sketch~\cite{Dekking1982}, but the difference is only superficial: our sketch connects the centre points of the tiles, while Dekking's sketch connects the entry and exit points of the tiles.}
\label{fig:dekking}
\end{figure}

\begin{theorem}\label{th:dekking}
Dekking's curve (as defined in Figure~\ref{fig:dekking}(a)) is a square curve with Arrwwid number three.
\end{theorem}
\begin{proof}
First we establish the entry and exit points of the unit square \unittile. The locations of entry and exit points are determined by the transformations of the recursive rule in the different tiles. The curve drawn in Figure~\ref{fig:dekking}(a) does not define the locations of entry and exit points; it is just a sketch for illustration. To prove that the entry and exit points are actually at the locations which the sketch may suggest, we have to study the way in which the tiles are transformed. The first subtile of the unit square to be visited is the one in the lower left corner. Since the recursive tiling and ordering rule \textsf{R} is applied to that tile without any rotation, reflection or reversal, the first subtile to be visited within it is again in the lower left corner, and this continues to hold to any depth of recursion. Hence, as $y$ approaches 0 from above, the fragment $\unittile[0,y]$ shrinks to the lower left corner point of \unittile, and this is therefore the entry point of \unittile. The last subtile of \unittile to be visited is the one in the lower right corner. Its exit point is what would be its entry point before the reversal transformation is applied; as a result of the rotation, this point is found in the lower right corner. Therefore the exit point of \unittile is its lower right corner.

Given any disk $Q$ of radius $r$, consider the square grid ${\cal T}$ that is formed by subdividing \unittile recursively until the squares have width at most $10r$ and more than $2r$. Now $Q$ intersects at most one horizontal line and at most one vertical line of~${\cal T}$. If at most one grid line intersects~$Q$, then $Q$ lies in at most two tiles of ${\cal T}$, with total area at most $200r^2 = \frac{200}\pi \area(Q)$. Otherwise $Q$ is intersected by a horizontal grid line and a vertical line, which intersect in a vertex~$p$, and $Q$ intersects the four tiles of ${\cal T}$ that share vertex~$p$. Below we prove that two of these tiles are consecutive in the scanning order, so that these four tiles constitute at most three space-filling curve fragments with total area at most $400r^2 = \frac{400}\pi \area(Q)$. This will prove that the space-filling curve has Arrwwid number at most three; together with the lower bound of Theorem~\ref{th:2dsfclb} in the next subsection this proves that the Arrwwid number is in fact exactly three.

For our proof that around any vertex $p$ two of the tiles meeting in $p$ are consecutive in the scanning order, we use the same terminology as in the proof of Theorem~\ref{th:Dauntiling}: a level-$k$ feature (vertex, edge, or tile) is one that first appears when expanding the recursion to a depth of $k$ levels down from the unit tile. Now let $k$ be the level of $p$, that is, $p$ is a level-$k$ vertex. Now $p$ is of one of two types: it either lies in the interior of a level-$(k-1)$ tile, or it lies on an edge between two level-$(k-1)$ tiles.

If $p$ is of the first type, one may verify in Figure~\ref{fig:dekking}(a) (see the vertices marked by squares) that two of the level-$k$ tiles around $p$, say $A$ and $B$, connect in $p$. This implies that if we expand the recursion further until we obtain ${\cal T}$, the tile of ${\cal T} \cap A$ that touches $p$ connects to the tile of ${\cal T} \cap B$ that touches $p$.

If $p$ is of the second type, that is, it lies on an edge between two level-$(k-1)$ tiles, then let $q$ be the corner shared by these two tiles which is closest to $p$. Now consider $p$ as a vertex on the boundary of the level-$(k-1)$ tile that lies to the right when walking from $p$ to $q$, see Figure~\ref{fig:dekking}(b). When Figure~\ref{fig:dekking}(a) depicts this tile, then $p$ is one of the vertices marked by black dots (note that here we exploit the fact that the recursive rule for this curve does not use reflections, so that we preserve the left/right orientations for all tiles). Again, one may verify in Figure~\ref{fig:dekking}(a) that two level-$k$ tiles connect in $p$, and hence two tiles of ${\cal T}$ connect in $p$.
\end{proof}

\begin{figure}
\centering
\includegraphics[width=\hsize]{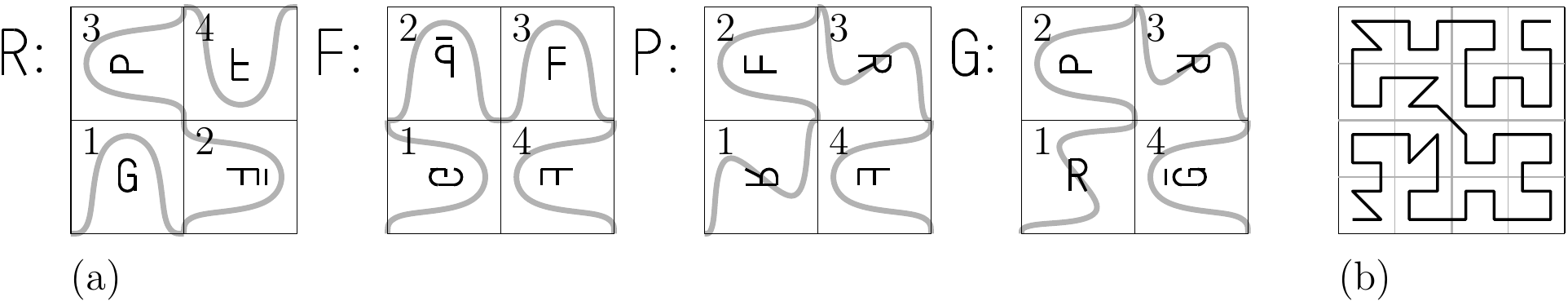}
\caption{(a) Definition of the \ARRWW scanning order. (b) Sketch of the level-three expansion (starting from rule \textsf{R}).}
\label{fig:arrww}
\end{figure}

\begin{figure}
\centering
\includegraphics[width=\hsize]{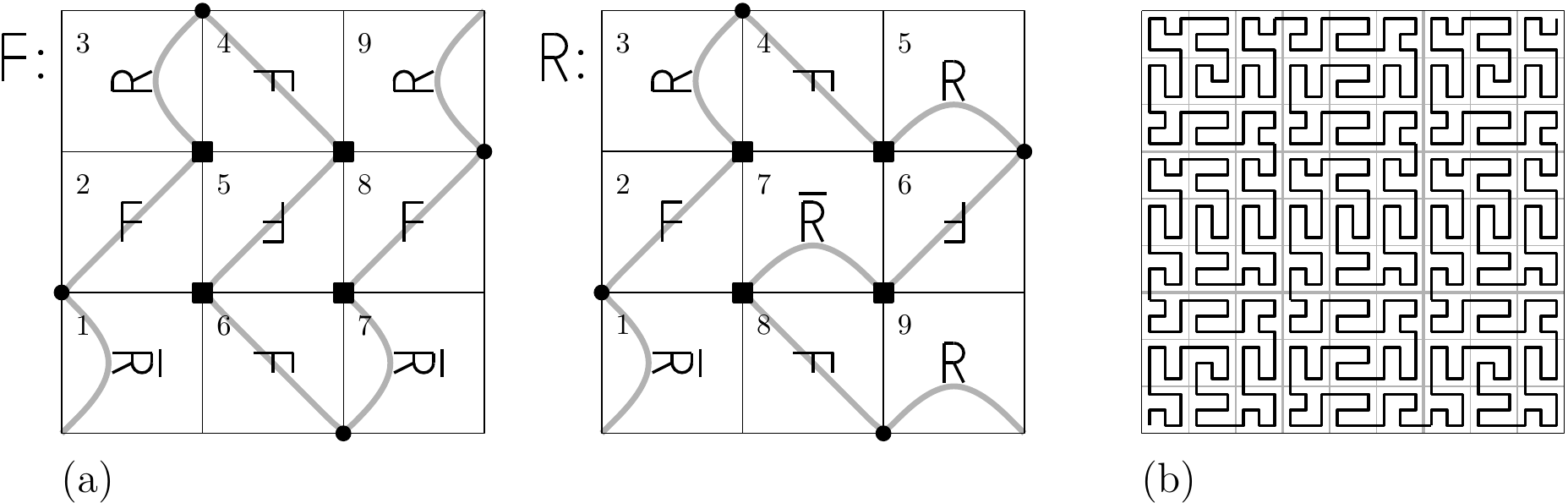}
\caption{(a) Definition of the Kochel scanning order. (b) Sketch of the level-three expansion (starting from rule \textsf{F}).}
\label{fig:kochel}
\end{figure}

The curve analysed above was described by Dekking in the form of a recursive polygonal approximation~\cite{Dekking1982} (he did not give a tiling-based definition or analyse its Arrwwid number). There are other uniform square curves with Arrwwid number three, with smaller tilings (less squares per square), but these are composite curves, not simple curves like Dekking's. Asano et al.\ define the \ARRWW-curve~\cite{Asano1997} (see Figure~\ref{fig:arrww}) and prove that it has Arrwwid number three. The \emph{Kochel curve} (Figure~\ref{fig:kochel}) is another curve with Arrwwid number three, which can be proven to have Arrwwid number three in a way very similar to the proof for Dekking's curve. Unlike the \ARRWW-curve and Dekking's curve, the Kochel curve has the property that consecutive tiles in the order always share an edge.
%The Kochel curve has slightly better worst-case bounding box quality values: the \ARRWW curve has WBA 3.06 and WBP 3.13; the Kochel curve has WBA 2.63 and WBP 2.89 (see Haverkort and Van Walderveen~\cite{Haverkort2009} for the definitions of these measures and how to calculate them).

\enlargethispage\baselineskip
In contrast to the curves mentioned above, the following square curves all have Arrwwid number four: Hilbert's curve~\cite{Hilbert1891}; Z-order, also known as Lebesgue's curve~\cite{Lebesgue1904}; the $\beta\Omega$-curve~\cite{Wierum2002}; Peano's curve~\cite{Peano1890} and all other simple uniform square curves of size nine, such as Luxburg's variations~\cite{Luxburg1998}, R-order, and Meurthe order (see Haverkort and Van Walderveen~\cite{Haverkort2009} for definitions in our notation).

\subsection{A lower bound}
Above we saw that there are uniform square space-filling curves with Arrwwid number three, even while the underlying recursive tiling has Arrwwid number four. For the specific case of square space-filling curves with four tiles per recursive rule, Asano et al.\ proved that this is optimal: the Arrwwid number cannot be less than three~\cite{Asano1997}. Below we see a different proof technique (which we can also generalise to three dimensions later) and prove that an Arrwwid number of three is in fact optimal for \emph{any} space-filling curve that is based on a recursive tiling whose tiles are topologically equivalent to disks---note that this category includes all tilings we have seen so far.

\begin{theorem}\label{th:2dsfclb}
Any space-filling curve based on a recursive tiling with tiles that are topologically equivalent to disks has Arrwwid number at least three.
\end{theorem}
\begin{proof}
Consider any subdivision of a space-filling curve, filling a unit tile \unittile, into a set $F'$ of $k$ fragments $\unittile[0,x_1], \unittile[x_1,x_2], ..., \unittile[x_{k-1},1]$, such that each fragment is a simply connected region in the plane, that is, topologically equivalent to a disk. Let $\overline{\unittile}$ be the unit tile's complement $\Reals^2 \setminus \unittile$, and let $F$ be $F' \cup \{\overline{\unittile}\}$.

The boundaries of the fragments in $F'$ form a plane graph ${\cal G}$ with face set~$F$. As the edge set $E$ of ${\cal G}$ we take the maximal curves that form a boundary between two faces of $F$. The vertex set $V$ of ${\cal G}$ is the set of points where three or more faces of $F$ meet. By $\degr(v)$ we denote the number of edges of $E$ that are incident on the vertex $v$. By $\doors(v)$ we denote the number of fragments of $F'$ that have $v$ as their entry point plus the number of fragments of $F'$ that have $v$ as their exit point.

\begin{figure}
\centering
\includegraphics{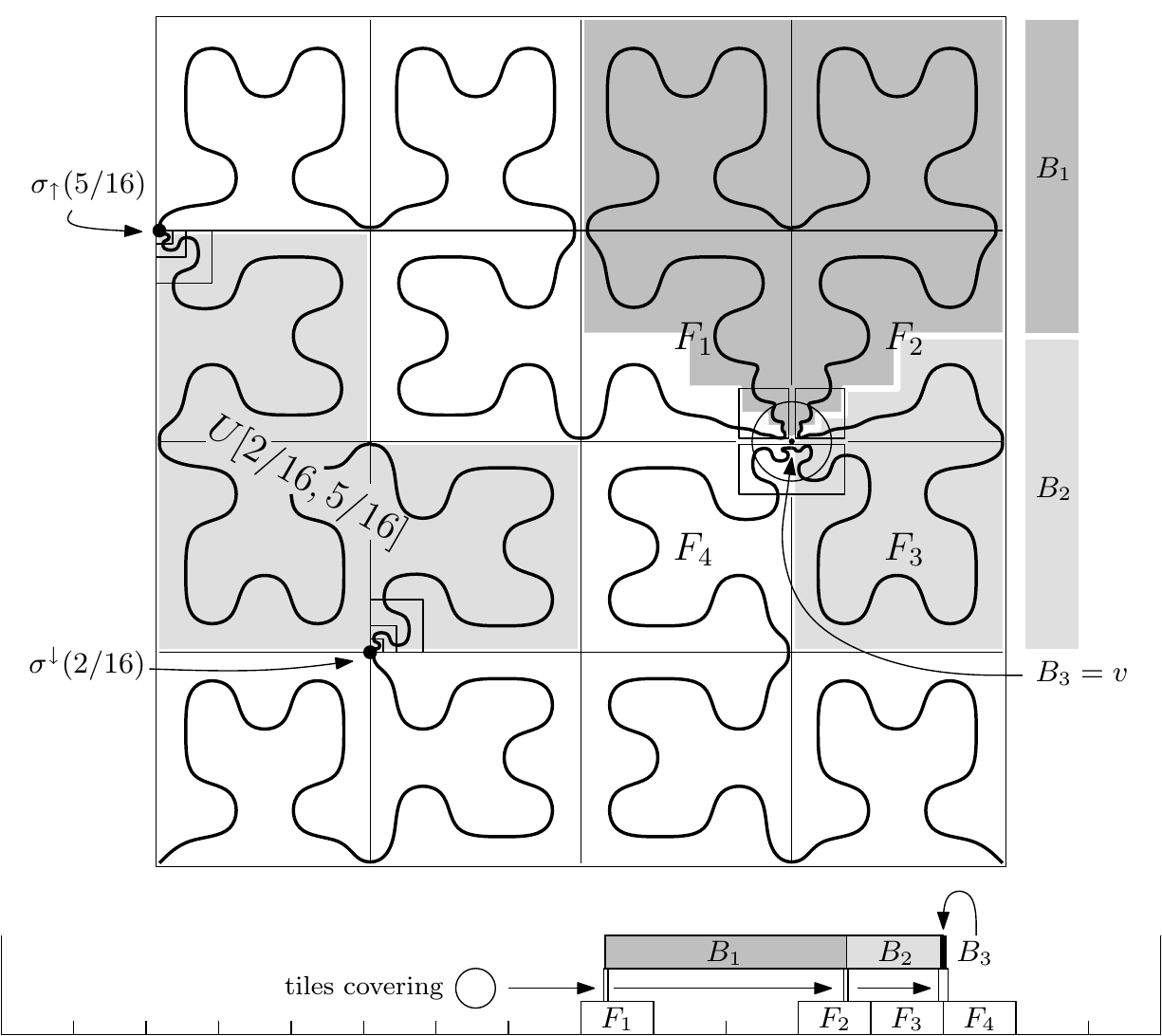}
\caption{Illustrating the proof of Theorem~\ref{th:2dsfclb}. On top we see a unit square divided into sixteen tiles, and below we see the corresponding unit interval that represents the space-filling curve, divided into sixteen fragments. The four tiles meeting in $v$ are $F_1, F_2, F_3$ and $F_4$, and the bridges around $v$ are $B_1$ (dark shaded area), $B_2$ (light shaded area), and $B_3$ (just the point $v$). Any disk $Q$ centered on $v$ can be covered with three small fragments (two squares, and one rectangle consisting of two squares). To do with fewer fragments, that is, only two, these fragments need to cover at least one of the bridges $B_1$ and $B_2$.}
\label{fig:bridges}
\end{figure}

Let $V'$ be the vertices of ${\cal G}$ that are not on the outer face $\overline{\unittile}$. I claim that if the space-filling curve has Arrwwid number less than three, we must have $\degr(v) - \frac12 \doors(v) \leq 2$ for every vertex $v \in V'$. To prove this, let's assume, for the sake of contradiction, that there is a vertex $v \in V'$ with $\degr(v) - \frac 12 \doors(v) > 2$. Let $F_1, ..., F_m$ be the faces of $F$ that meet at $v$ (where $m = \degr(v)$), in the order in which they appear in the scanning order; each face $F_i$ constitutes a fragment $U[f_i,f'_i]$. For $1 \leq i < m$, let $B_i$ (bridge $i$) be the smallest fragment $U[b_i,b'_i]$ (not necessarily a fragment in $F'$) that starts from $v$ in $F_i$ and ends at $v$ in $F_{i+1}$, more precisely: $B_i$ is the smallest fragment $U[b_i,b'_i]$ such that $b_i \in [f_i,f'_i]$, $b'_i \in [f_{i+1},f'_{i+1}]$, and $\sigma^{\downarrow}(b_i) = \sigma_{\uparrow}(b'_i) = v$ (see Figure~\ref{fig:bridges} for an example). Note that $B_i$ degenerates to the point $v$ if and only if $F_i$ and $F_{i+1}$ connect at $v$, contributing two to the entry and exit points counted in $\doors(v)$. Therefore the number of degenerate bridges at $v$ is at most $\frac12 \doors(v)$. Since $\frac 12 \doors(v) < \degr(v) - 2$ and there are $\degr(v)-1$ bridges in total, there are at least two non-degenerate bridges at $v$. By assumption the space-filling curve has Arrwwid number less than three, so there must be a constant $c$ such that any circular query range $Q$ can be covered by at most two fragments (not necessarily from $F'$) with total area at most $c \cdot \area(Q)$. Now consider a circle $Q$ centered at $v$ with area less than $1/c$ times the area of the smallest non-degenerate bridge at $v$. To cover $Q$ with only two fragments, the fragments must cover $Q$ and all but one of the bridges $B_1, ..., B_{m-1}$. Thus at least one non-degenerate bridge is covered, and by the definition of $Q$, this bridge has area more than $c \cdot \area(Q)$. However, this contradicts the definition of $c$. Therefore, if the space-filling curve has Arrwwid number less than three, we must have $\degr(v) - \frac12 \doors(v) \leq 2$ for every vertex $v \in V'$.

Summing the above claim over all vertices in $V'$ we get:\begin{equation}\label{eq:vertexsum}
\sum_{v \in V'} \degr(v) \leq 2|V'| + \frac12 \sum_{v \in V'} \doors(v) < 2|V| + \frac12 \sum_{v \in V'} \doors(v).
\end{equation}
Let $E'$ be the set of edges with both end points in $V'$, and let $E''$ be the remaining edges (those with at least one end point on the outer face). We have $\sum_{v \in V'} \degr(v) > 2|E'| = 2|E| - 2|E''|$ and $\sum_{v \in V'} \doors(v) \leq 2|F| - 2$ (since each face except the outer face has one entry point and one exit point). Thus Equation~\ref{eq:vertexsum} implies:\[
2|E| - 2|E''| < 2|V| + |F| - 1.
\]
Moving $2|E''|$ and $F$ to the other side and subtracting Euler's formula $|E| - |F| = |V| - 2$ twice we get:\begin{equation}\label{eq:maxfaces}
|F| < 2|E''| + 3.
\end{equation}

Now take any space-filling curve based on a recursive tiling whose tiles are simply connected regions (topological disks). Consider the subdivision of the unit tile \unittile into tiles down to a level of recursion on which there is a tile $T$ that does not touch the boundary of \unittile. This level must exist, since the maximum diameter of the tiles decreases geometrically with the level of recursion, so at some point the tiles become so small that the tile that covers the centre point of \unittile does not touch the boundary of~\unittile. Now subdivide $T$ into smaller tiles recursively while keeping the tiles outside $T$ as they are. Thus we keep $E''$ fixed while we continue to increase $|F|$, eventually getting $|F| \geq 2|E''| + 3$. This contradicts Equation~\ref{eq:maxfaces}, so the space-filling curve cannot have an Arrwwid number less than three.
\end{proof}

A direct consequence of the above is that when a space-filling curve is based on a recursive tiling with simply connected tiles and Arrwwid number three, the Arrwwid number of the space-filling curve is always exactly three as well, regardless of the order in which the curve traverses the tiles; the ``curve'' does not even need to be continuous.

\section{Three-dimensional tilings}
The definitions of tilings and Arrwwid numbers generalise naturally to higher dimensions. The Arrwwid number of a three-dimensional recursive tiling of a unit tile \unittile is the smallest number $a$ such that there is a constant $c$ such that any ball $Q$ that lies entirely in \unittile can be covered by $a$ tiles with total volume at most $c \cdot \vol(Q)$. Theorem~\ref{th:2dtilinglb} generalises naturally to three dimensions:
\begin{theorem}\label{th:3dtilinglb}
Each recursive tiling of a three-dimensional region \unittile has Arrwwid number at least four
(proof omitted).
\end{theorem}

\subsection{Optimal recursified tilings}\label{sec:3drecursified}

\begin{figure}
\centering
\includegraphics[height=7cm]{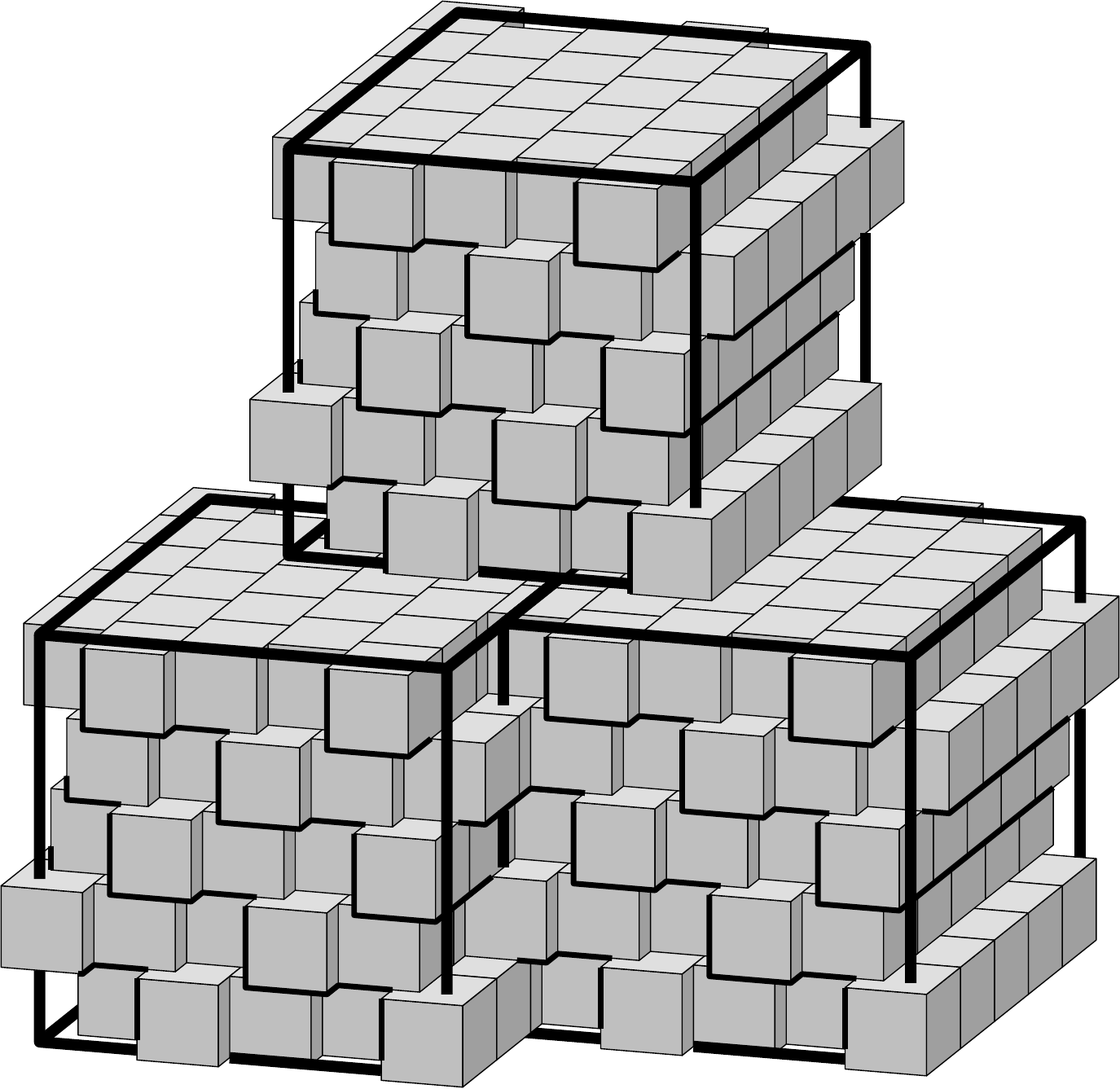}
\caption{Approximating large cubes in a tiling with vertex degree four by 125 small cubes each, taken from the same tiling scaled by a factor 1/5.}
\label{fig:shiftedcubes}
\end{figure}

\label{sec:3dshiftedcubes}
This subsection describes how to construct a recursified three-dimensional tiling with optimal Arrwwid number, that is, four. The construction is based on a tiling with cubes that are shifted with respect to each other (Figure~\ref{fig:shiftedcubes}). The coarse tiling consists of horizontal layers that are one cube high; each layer is shifted to the right and to the front with respect to the layer below over 1/3 of a cube's width. Each layer consists of columns that are one cube wide; each column is shifted to the back with respect to the column to the left over 1/3 of a cube's width. The fine tiling is equal to the coarse tiling scaled by a factor 1/5 (with the centre point of a large cube as its fixed point); each small tile is assigned to the large tile with which it has the largest overlap.

\begin{theorem}
There is a recursified three-dimensional tiling with Arrwwid number four.
\end{theorem}
\begin{proof}
The proof follows the approach of Theorem~\ref{th:gosper}. The recursified tiling is constructed from the coarse and fine tilings of shifted cubes as described above. Let $w$ be the width of a cube in the coarse tiling, and thus, $w' = w/5$ is the width of a cube in the fine tiling. In the coarse tiling, the smallest ball that intersects more than four tiles has radius $w/6$. When we replace a large tile by the union of 125 small tiles, the boundary of the large tile stays within a distance of $\frac13 \sqrt{2} \cdot w' = \frac{1}{15}\sqrt{2} \cdot w$ from its original location. In recursion with scale factor 1/5, the movement of the boundary adds up to at most $\frac54 \cdot \frac1{15}\sqrt{2} \cdot w = \frac{1}{12}\sqrt{2} \cdot w$. Thus the smallest ball that intersects more than four large tiles will still have radius at least $(\frac16 - \frac1{12}\sqrt{2})w > \frac{1}{21}w$. The proof can now be completed as in Theorem~\ref{th:gosper}.
\end{proof}

Another way to get this result is based on a regular tiling with truncated octahedra. In such a tiling all vertices are incident on four tiles. The coarse tiling contains the truncated octahedron whose vertices have coordinates $(0, \plmin 5, \plmin 10)$, and all permutations of these coordinates. It is the intersection of an axis-parallel cube with diagonal $(-10,-10,-10)-(10,10,10)$, and an octahedron with vertices $(\plmin 15,0,0)$, $(0,\plmin 15,0)$, and $(0,0,\plmin 15)$. Further truncated octahedra are placed at translations $(20k,20l,20m)$ and $(20k+10,20l+10,20m+10)$, for all $k,l,m \in \Integers$. The fine tiling is the same as the coarse tiling scaled by a factor 1/5. To the large tile centered at the origin we assign 125 small tiles, namely those with centre points at the following coordinates or permutations thereof: $(0,0,0)$, $(0,0,\plmin 4)$, $(\plmin 2,\plmin 2,\plmin 2)$, $(0,\plmin 4,\plmin 4)$, $(0,0,\plmin 8)$, $(\plmin 2,\plmin 2, \plmin 6)$, $(\plmin 4,\plmin 4,\plmin 4)$, $(\plmin 2,\plmin 6,\plmin 6)$, $(\plmin 0,\plmin 4,\plmin 8)$, $(2,2,10)$, $(-2,-2,10)$, and $(-2,2,-10)$. %Figure~\ref{fig:octahedron}(b) shows the resulting approximation of the large tile.
To prove that the resulting tiling has Arrwwid number four, we need the following. In the coarse tiling, the smallest ball that intersects more than four tiles has radius $\frac52\sqrt 2$. When we replace a large tile by a union of small tiles, the boundary of the large tile stays within a distance of 2 from the original boundary. In recursion with scale factor 1/5, the movement of the boundary adds up to at most $\frac54 \cdot 2 = \frac52$. Thus the smallest ball that intersects more than four large tiles will still have radius at least $\frac52 (\sqrt{2} - 1)$. The proof can now be completed as in Theorem~\ref{th:gosper}.

\subsection{Rectangular tilings}

The following is a straightforward generalisation of Observation~\ref{ob:squaretiling}.
\begin{observation}\label{ob:cubetiling}
Any uniform cube tiling has Arrwwid number eight.
\end{observation}

\begin{figure}
\centering
\includegraphics[width=7cm]{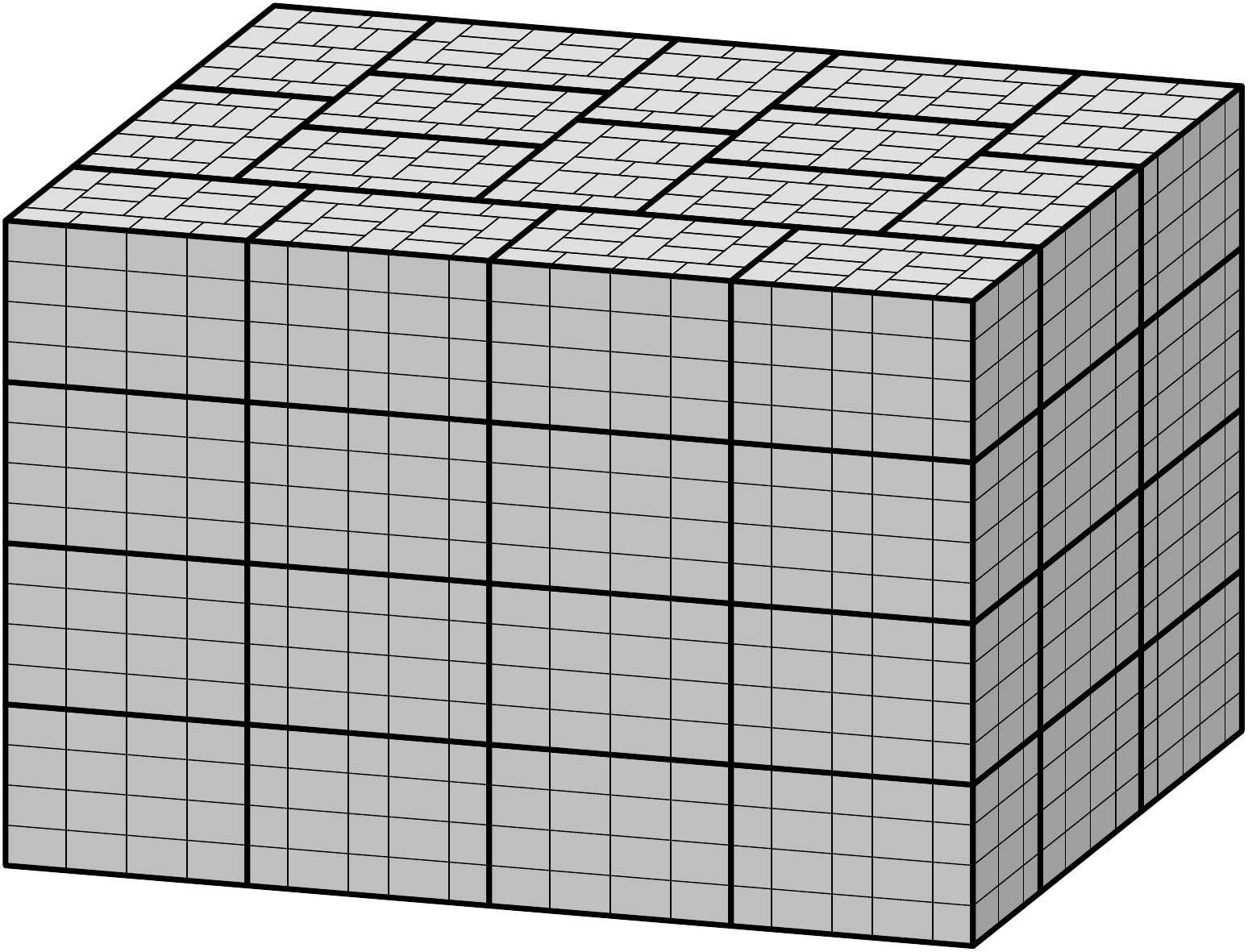}
\caption{Lifting a Daun tiling into three dimensions.}
\label{fig:lifted-daun}
\end{figure}

\label{sec:3dliftedDaun}
With rectangular tiles, that is, with axis-parallel boxes, one can do better. The three-dimensional rectangular tiling with lowest Arrwwid number found so far, is the ``lifted Daun tiling'' shown in Figure~\ref{fig:lifted-daun}. It is obtained by adding a third dimension to the two-dimensional Daun tiling from Figure~\ref{fig:daun}. The unit tile of the three-dimensional tiling is a rectangular block with width-to-depth ratio 3/2 and arbitrary height. It is divided into 64 subtiles, organised into 4 equal layers of 16 tiles each. Each layer, seen from above, shows a Daun tiling. The tiles are rotated around a vertical axis only in the same way as in the Daun tiling. Each vertex in the resulting tiling is on the boundary between two layers and is adjacent to three tiles in each layer (because the Daun tiling has vertex degree three), therefore the vertex degree of the three-dimensional tiling is six. From here one can prove (similar to Theorem~\ref{th:Dauntiling}) that the three-dimensional ``lifted Daun tiling'' has Arrwwid number six.

\begin{theorem}\label{th:lifteddauntiling}
There is a three-dimensional rectangular recursive tiling with Arrwwid number six.
\end{theorem}

We do not know if it is possible to get an Arrwwid number of four or five with a rectangular tiling in three dimensions. As in the two-dimensional case, we could derive some properties that the aspect ratios and the numbers of tiles should have in order to allow an Arrwwid number of four. Nevertheless the search space is still huge. So far I only managed to search all three-dimensional uniform rectangular tilings with less than 27 tiles; no recursive tiling with Arrwwid number four was found.

\section{Three-dimensional space-filling curves}

\subsection{Rectangular space-filling curves}

\begin{theorem}\label{th:cubesfc}
Any space-filling curve based on a uniform cube tiling has Arrwwid number at least seven.
\end{theorem}
\begin{proof}
Consider the regular cube tiling ${\cal T}$ obtained by subdividing a unit cube \unittile recursively into smaller cubes to a certain depth of recursion. Let $C$ be the set of tiles obtained. By $\tdegr(v)$ we denote the number of tiles in $C$ that are incident on the vertex $v$. By $\doors(v)$ we denote the number of tiles in $C$ that have $v$ as their entry point plus the number of tiles in $C$ that have $v$ as their exit point.

Let $V'$ be the set of vertices of ${\cal T}$ that do not lie on the boundary of \unittile. If the space-filling curve has Arrwwid number less than seven, we must have $\tdegr(v) - \frac12 \doors(v) \leq 6$ for every vertex $v \in V'$ (the proof of this claim is a straightforward adaptation of the proof of Theorem~\ref{th:2dsfclb}). Note that $\tdegr(v) = 8$ for each $v \in V'$, so in fact we must have $\frac12 \doors(v) \geq 2$ for each $v \in V'$. Since the total number of entry and exit points at vertices is at most $2|C|$, this leads to:\begin{equation}\label{eq:cubesfc}
|C| \geq \frac12 \sum_{v \in V} \doors(v) \geq \sum_{v \in V'} \frac12 \doors(v) \geq 2|V'|.
\end{equation}
However, in a uniform cube tiling we have $|V'| = (|C|^{1/3} - 1)^3$, which is more than $\frac12 |C|$ for a sufficiently deep level of recursion. This contradicts Equation~\ref{eq:cubesfc}, so the space-filling curve cannot have an Arrwwid number less than seven.
\end{proof}

\begin{theorem}
Any space-filling curve based on the lifted Daun tiling of Section~\ref{sec:3dliftedDaun} has Arrwwid number six.
\end{theorem}
\begin{proof}
The proof is similar to that of Theorem~\ref{th:cubesfc}. In this case a contradiction is derived from $\tdegr(v) - \frac12 \doors(v) \leq 5$ and $\tdegr(v) = 6$ for each $v \in V'$, which gives $\frac12 \doors(v) \geq 1$ for each $v \in V'$. Summing up over all $v \in V'$ we get $|C| \geq |V'|$. Now consider the tiling obtained by expanding one level of recursion of the lifted Daun tiling. Whenever we subdivide a tile in this tiling, we replace a tile by 64 smaller tiles, and add a number of new vertices: to start with, on each of the three boundaries between the four layers of new tiles, there are 30 new vertices (see Figure~\ref{fig:daun}), at least 24 of which lie in the interior of the unit tile \unittile (since at least one short side and at least one long side of the tile lies in the interior of \unittile). Thus we increase the number of tiles by 63 while increasing the number of vertices in the interior of \unittile by at least 72. Therefore, when we subdivide enough tiles, we get $|V'| > |C|$, which contradicts the above. Therefore the space-filling curve cannot have Arrwwid number less than six.
\end{proof}

\subsection{A lower bound for convex space-filling curves}

To prove a lower bound on the Arrwwid number of three-dimensional space-filling curves, we follow the same general approach as in two dimensions. The proof in the two-dimensional case essentially has two main ingredients: (1) In any subdivision of a two-dimensional unit tile into simply connected tiles, the total number of vertex-tile incidences is roughly twice the number of edges; therefore Euler's formula gives us a relation ${\cal A}$ between the number of vertices, the number of tiles, and the total number of vertex-tile incidences. (2) For the Arrwwid number to be lower than three, there must be another relation ${\cal B}$ between the number of vertices, the number of tiles, and the total number of vertex-tile incidences. It turns out ${\cal A}$ conflicts with ${\cal B}$ for large enough tilings, hence an Arrwwid number lower than three is impossible.

Unfortunately ingredient (1) does not easily generalise to three dimensions, because besides the number of edges, another unknown enters the equations, namely the number of two-dimensional facets between the tiles. Therefore, to be able to complete our proof, we need another way to establish a relation between the number of vertices, the number of tiles, and the number of vertex-tile incidences in a three-dimensional tiling. We will now see how this can be done when the tiles are convex polyhedra.

\begin{figure}
\centering
\includegraphics[width=\hsize]{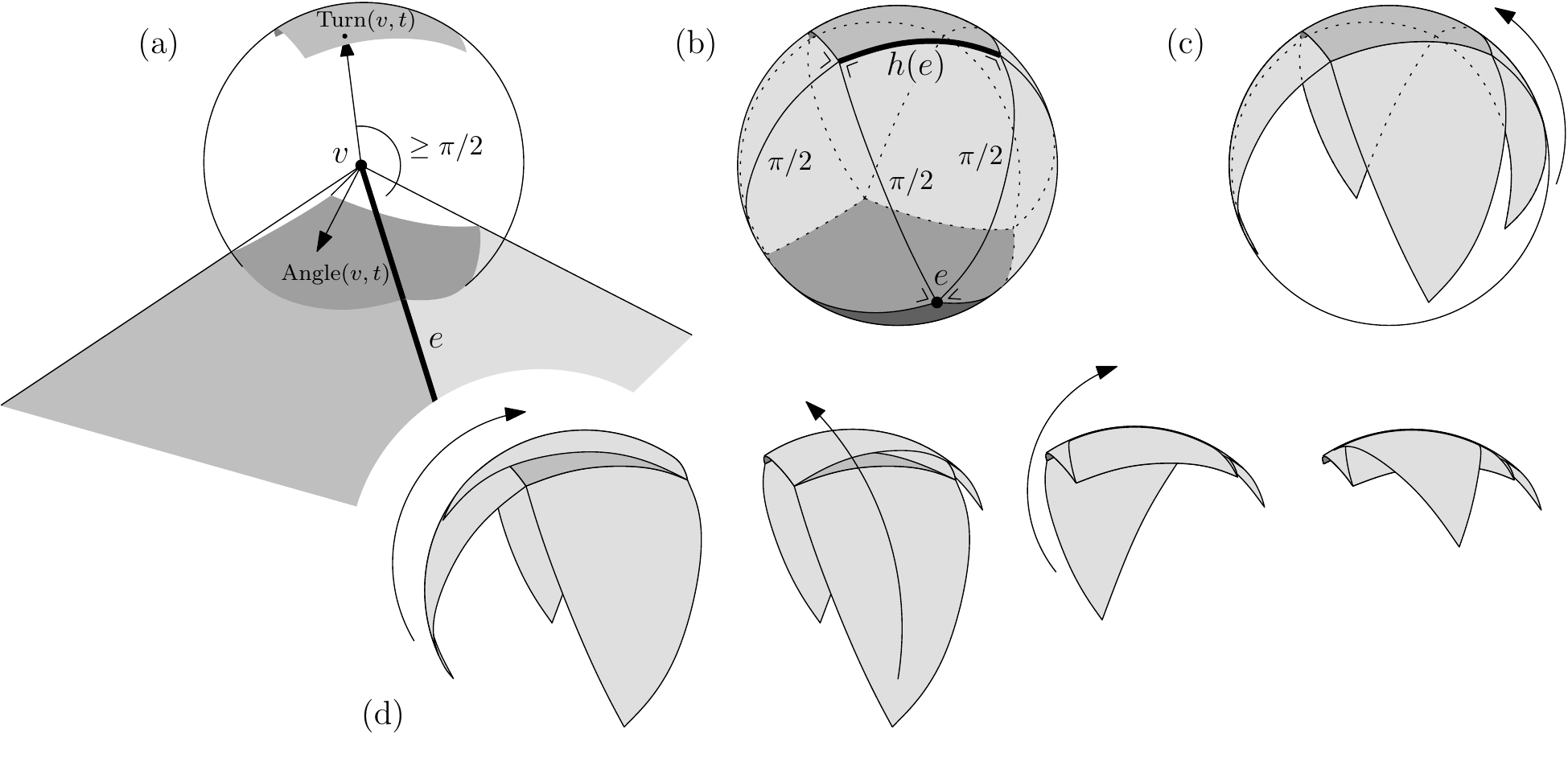}
\caption{For a vertex $v$ on a polyhedron $t$, we have that $\Angl(v,t)$ and $\Turn(v,t)$ can be covered by the triangles between them on the unit sphere of directions.}
\label{fig:3d-turn}
\end{figure}

Consider a convex polyhedron $t$, and let $v$ be a point on the boundary of $t$. In the following, we define the size of a set of vectors $S$ as the area of the projection of the vectors in $S$ on a unit sphere. Let $\angl(v,t)$ be the size of the set $\Angl(v,t)$, which consists of the vectors that point from $v$ into $t$. Let $\turn(v,t)$ be the size of the set $\Turn(v,t)$, which consists of the vectors that point from $v$ away from $t$ at an angle of at least $\pi/2$ with the boundary of $t$. Figure~\ref{fig:3d-turn}(a) illustrates the case in which $v$ is a vertex of $t$. Note that if $v$ is a point in the interior of an edge of $t$, we have $\angl(v,t) \in (0,2\pi)$ and $\turn(v,t) = 0$; if $v$ is a point in the interior of a facet of~$t$, we have $\angl(v,t) = 2\pi$ and $\turn(v,t) = 0$. We now have the following:
\begin{lemma}\label{le:angleplusturn}
$\turn(v,t) + \angl(v,t) \leq 2\pi$. Equality holds if and only if $v$ is a point in the interior of a facet of $t$, in which case $\angl(v,t) = 2\pi$ and $\turn(v,t) = 0$.
\end{lemma}
\begin{proof}
Let $e_1,...,e_m$ be the edges incident to $v$. For any edge $e_i$, let $h(e_i)$ be the plane through $v$ that is orthogonal to $e_i$, and let $H(e_i)$ be the halfspace bounded by $h(e_i)$ that does not contain~$e_i$. Observe that the vectors that point from $v$ away from $t$ at an angle of at least $\pi/2$ are exactly those that point into the intersection of the halfspaces $H(e_1),...,H(e_m)$. Thus, on the unit sphere of directions, we find that $\Angl(v,t)$ is represented as a spherical polygon whose vertices correspond to the directions of the edges $e_1,...,e_m$ with respect to $v$, and $\Turn(v,t)$ is represented as a spherical polygon whose edges are segments of $h(e_1),...,h(e_m)$, see Figure~\ref{fig:3d-turn}(b). The relation between $e_i$ and $h(e_i)$ as explained above implies that the remaining part of the unit sphere, between $\Angl(v,t)$ and $\Turn(v,t)$, can be triangulated by arcs of length $\pi/2$ that make right angles with the adjacent edges of $\Angl(v,t)$ and $\Turn(v,t)$.

Now consider $\Turn(v,t)$ and the adjacent triangles, see Figure~\ref{fig:3d-turn}(c). Since $\Turn(v,t)$ is smaller than a hemisphere, the triangles would cover $\Turn(v,t)$ completely if folded onto it, see Figure~\ref{fig:3d-turn}(d). Similarly, the triangles adjacent to $\Angl(v,t)$ would cover $\Angl(v,t)$ completely. Hence the total area of the triangles is at least $\angl(v,t) + \turn(v,t)$, so $\angl(v,t) + \turn(v,t)$ is at most $2\pi$, the area of half a unit sphere.

Equality holds only if the triangles adjacent to $\Turn(v,t)$ cover $\Turn(v,t)$ \emph{exactly}, and those adjacent to $\Angl(v,t)$ cover $\Angl(v,t)$ \emph{exactly}. This only happens when both $\Angl(v,t)$ and $\Turn(v,t)$ are either empty, or a full hemisphere. $\Angl(v,t)$ is never empty, but it is possible that $\angl(v,t) = 2\pi$ and $\turn(v,t) = 0$, namely if $v$ is a point in the interior of a facet of $t$.
\end{proof}

Consider a subdivision of a bounded convex polyhedral unit tile $\unittile$ into a set $C$ of convex polyhedral tiles, whose shapes come from a fixed set $S$. Let $V$ be the set of vertices of the subdivision, and let $V' \subset V$ be the vertices that lie in the interior of $\unittile$. We can now prove:
\begin{lemma}\label{le:3dvertextileincidences}
$|V'| + (1+\alpha)|C| < \frac12 \sum_{v \in V} \tdegr(v)$, where $\tdegr(v)$ is the number of tiles with $v$ on their boundaries, and $\alpha > 0$ is a constant that depends on $S$.
\end{lemma}
\begin{proof}
For ease of notation, define $\angl(v,t) = \turn(v,t) = 0$ when $v$ does not lie on the boundary of $t$. Observe that for every vertex $v \in V'$ the sets $\Angl(v,t)$ of all incident tiles $t$ together cover the full sphere of directions, so we have $\sum_{t \in C} \angl(v,t) = 4\pi$. Also, for every tile $t \in C$ the sets $\Turn(v,t)$ of all vertices $v$ on $t$ together cover the full sphere of directions, so we have $\sum_{v \in V} \turn(v,t) = 4\pi$. With these observations we get:
\begin{eqnarray*}
& |V'| + |C|
= \\
& \frac1{4\pi} \sum_{v \in V'} \sum_{t \in C} \angl(v,t) + \frac1{4\pi} \sum_{t \in C}  \sum_{v \in V} \turn(v,t) < \\
& \frac1{4\pi} \sum_{v \in V} \sum_{t \in C} \left(\angl(v,t) + \turn(v,t)\right).
\end{eqnarray*}
Since every tile $t \in C$ is a bounded convex polyhedron with a shape from a fixed set $S$, it must have at least one vertex $v$ with $\turn(v,t) > 0$. Hence, by Lemma~\ref{le:angleplusturn}, for this vertex $v$ we have $\angl(v,t) + \turn(v,t) \leq 2\pi - \beta$, where $\beta > 0$ is a constant that depends on $S$. For other vertices~$v$ on the boundary of $t$ we have $\angl(v,t) + \turn(v,t) \leq 2\pi$ (by Lemma~\ref{le:angleplusturn}), and for vertices $v$ that are not on the boundary of $t$ we have $\angl(v,t) + \turn(v,t) = 0$ (by definition). Therefore the above equation gives:\[
|V'| + |C| <
\left(\frac1{4\pi} \sum_{v \in V} \tdegr(v)\cdot 2\pi\right) - \frac\beta{4\pi}|C| =
\left(\frac12 \sum_{v \in V} \tdegr(v)\right) - \frac\beta{4\pi}|C|.
\]
Setting $\alpha = \beta/4\pi$ proves the lemma.
\end{proof}

We are now ready to prove the lower bound we are after:

\begin{theorem}\label{th:3dsfclb}
Any space-filling curve based on a recursive tiling with convex tiles in three dimensions has Arrwwid number at least four.
\end{theorem}
\begin{proof}
Consider a subdivision ${\cal T}$ of the unit tile $\unittile$ into a set $C$ of convex polyhedral tiles, obtained by applying the recursive rules that define the tiling that underlies the space-filling curve. Let $V$ be the set of vertices of ${\cal T}$, and let $V' \subset V$ be the vertices that are not on the boundary of $\unittile$. Let $\tdegr(v)$ be the number of tiles in $C$ that are incident on the vertex $v$. Let $\doors(v)$ be the number of tiles in $C$ that have $v$ as their entry point plus the number of tiles in $C$ that have $v$ as their exit point.

If the space-filling curve has Arrwwid number less than four, we must have $\tdegr(v) - \frac12 \doors(v) \leq 3$ for every vertex $v \in V'$ (again, the proof of this claim is a straightforward adaptation of the proof of Theorem~\ref{th:2dsfclb}). Summing over all vertices in $V'$ gives:\[
\sum_{v \in V'} \tdegr(v) \leq 3|V'| + \frac12 \sum_{v \in V'} \doors(v).
\]
Let $V''$ be the set of vertices of ${\cal T}$ on the boundary of \unittile, that is, $V'' = V \setminus V'$. Now Lemma~\ref{le:3dvertextileincidences} gives $|V'| + (1+\alpha)|C| - \frac12 \sum_{v \in V''} \tdegr(v) < \frac12 \sum_{v \in V'} \tdegr(v)$ for some fixed constant $\alpha$. Because all tiles are convex polyhedra, all vertices $v \in V'$ have $\tdegr(v) \geq 4$, and therefore $2|V'| \leq \frac12 \sum_{v \in V'} \tdegr(v)$. Note that the total number of entry and exit points at vertices is at most $2|C|$. Therefore the above equation gives us:\[
3|V'| + (1+\alpha)|C| - \frac12 \sum_{v \in V''} \tdegr(v) < \sum_{v \in V'} \tdegr(v) \leq 3|V'| + \frac12 \sum_{v \in V'} \doors(v) \leq 3|V'| + |C|.
\]
Therefore we must have:\begin{equation}\label{eq:anysfc}
\alpha|C| < \frac12 \sum_{v \in V''} \tdegr(v).
\end{equation}
From here we follow the same approach as in the proof of Theorem~\ref{th:2dsfclb}: we construct ${\cal T}$ by subdividing the unit tile \unittile into tiles recursively until we get a tile $T$ that does not touch the boundary of \unittile. Then we subdivide $T$ further, keeping the vertices $V''$ and their incident tiles fixed, while continuing to increase $|C|$, eventually getting $\alpha|C| \geq \frac12 \sum_{v \in V''} \tdegr(v)$. This contradicts Equation~\ref{eq:anysfc}, so the space-filling curve cannot have Arrwwid number less than four.
\end{proof}

\section{More than three dimensions}
Many of the results for three dimensions, but not all, generalise to higher dimensions in a straightforward way. These results are given here without writing out the proofs:

\begin{theorem}
Each recursive tiling of a $d$-dimensional region \unittile has Arrwwid number at least $d+1$.
\end{theorem}
\begin{theorem}
There is a recursified $d$-dimensional tiling with Arrwwid number $d+1$.
\end{theorem}
(The tiling is constructed from a tiling of shifted hypercubes as in Section~\ref{sec:3dshiftedcubes}.)
\begin{theorem}
Any uniform hypercube tiling has Arrwwid number $2^d$, and any space-filling curve based on it has Arrwwid number $2^d - 1$ or $2^d$.
\end{theorem}
\begin{theorem}
There is a rectangular recursive tiling in $d$ dimensions with Arrwwid number $\frac34 \cdot 2^d$.
\end{theorem}
(A Daun tiling lifted into $d$ dimensions has Arrwwid number $3 \cdot 2^{d-2}$.)

\section{Discussion and conclusions}\label{sec:conclusions}

\paragraph{Arrwwid-optimal curves and tilings.} This paper shows that in two dimensions the lowest possible Arrwwid number is achieved by certain uniform square space-filling curves (such as the \ARRWW curve, the Kochel curve, and Dekking's curve) and by a certain rectangular tiling. However, for $d \geq 3$, no uniform cube space-filling curve can have an Arrwwid number as low as the best known rectangular space-filling curve, and the best known rectangular space-filling curve does not have an Arrwwid number as low as the best known space-filling curve on a fractal tiling. A three-dimensional rectangular recursive tiling with Arrwwid number four or five might exist, but so far, it was not found.

Our lower bounds for space-filling curves only apply to curves based on tilings with simply connected tiles (in two dimensions) or with convex tiles (in three dimensions). Lifting the restrictions on the shapes of the tiles in these bounds remains a topic for further research and I am still pursuing some new ideas on the matter. It may be interesting to investigate the relation between the Arrwwid number of a space-filling curve and its maximum multiplicity (the maximum number of times any point is visited by the curve). Although the Arrwwid number and the multiplicity of a curve may differ (Lebesgue's curve has multiplicity five and Arrwwid number four), it is possible that ideas related to the multiplicity of space-filling curves~\cite{Mandelbrot1983} prove useful in deriving bounds on Arrwwid numbers, at least for certain classes of curves.

\begin{figure}
\centering
\includegraphics[height=7.25cm]{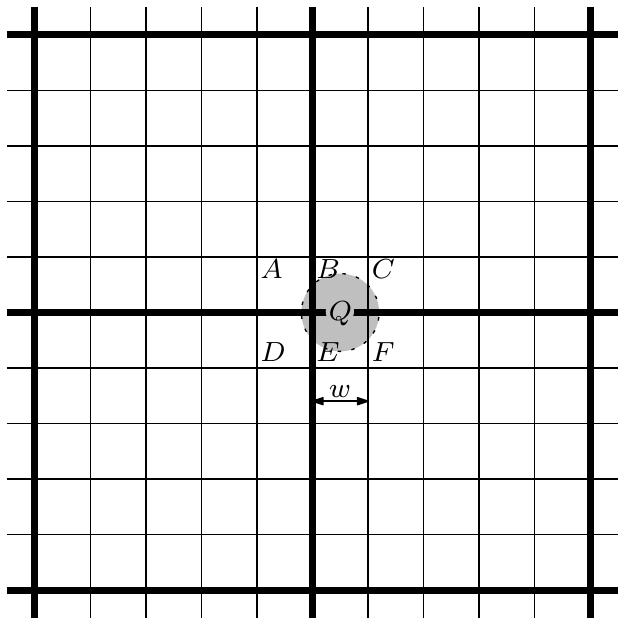}
\caption{A query range for which straightforward proofs lead to overly pessimistic bounds on the cover ratio of Dekking's curve.}
\label{fig:constants}
\end{figure}

\paragraph{Arrwwid numbers versus cover ratios.}
In principle all proofs of Arrwwid numbers come with an upper bound on the cover ratio $c$. Looking at the proofs in this paper one might get the impression that low Arrwwid numbers tend to come at the expense of extremely high cover ratios. This may be a shortcoming of the proofs rather than the space-filling curves. For example, consider Dekking's curve, which is based on recursively subdividing squares into 25 squares. Figure~\ref{fig:constants} shows a disk $Q$ that sticks out just a little bit from the tiles $B$ and $E$ with width $w$, also intersecting $A$, $C$, $D$ and $F$. The approach from the proof in Section~\ref{sec:squarewithA3} would be to cover $Q$ with the parents of these tiles. This results in a cover ratio bound of $100w^2 / \frac\pi4 w^2 = 400/\pi$. However, taking the scanning order into account (see Figure~\ref{fig:dekking}), we see that Dekking's curve has the property that if $B$ and $C$ are not adjacent in the order, then $E$ and $F$ must be adjacent, and vice versa. Furthermore, a corner tile and its neighbour are never far apart in Dekking's scanning order: there are at most three tiles between them. Note that among $A$, $B$, $D$ and $E$ there must also be pair that is consecutive in the scanning order. Therefore $Q$ can be covered with three fragments containing the six tiles $A$, $B$, $C$, $D$, $E$, and $F$, and at most three additional tiles. This results in a cover ratio of at most only $9w^2 / \frac\pi4 w^2 = 36/\pi$ (for this particular choice of $Q$). A more detailed case analysis that takes the scanning order into account may thus give much better bounds on the cover ratios of large tilings than the bounds presented in this paper.

\paragraph{Worst-case versus average case.}
Arrwwid numbers only consider the \emph{worst-case} number of tiles or curve fragments that are needed to cover a query range. Therefore lower Arrwwid numbers do not necessarily give better disk efficiency on average. It is possible that optimising the worst case actually has an adverse effect on the average-case performance. In the case of square curves of size four there is some intuition to support this concern: getting an Arrwwid number of three requires the use of diagonal connections in this case. For any tile $A$, let $X(A)$ be the set of query ranges $Q$ such that $Q$ intersects $A$, but not the next tile of the same size in the scanning order. The total of $|X(A)|$ over all tiles $A$ gives an indication of the probability that after scanning a tile, we must either move the disk head to the next tile that intersects the query range, or we scan a tile with false answers only. Since $X(A)$ is larger when $A$ is connected diagonally to the next tile than when it is connected orthogonally, diagonal connections may lead to decreased performance on average.

\paragraph{Theory versus practice.}
In two dimensions, most widely known curves have Arrwwid number four (Peano's curve~\cite{Peano1890}, Hilbert's curve~\cite{Hilbert1891}, Lebesgue/Z-order~\cite{Lebesgue1904}, Sierpi\'nski/Knopp/H-order~\cite{Niedermeier2002,Sagan1994}), while Arrwwid-optimal curves have Arrwwid number three (the \ARRWW-curve, the Kochel curve, and Dekking's curve, all discussed in Section~\ref{sec:squarewithA3}). One may wonder if there are any practical settings in which it matters much whether three of four fragments are used to cover a query range in the worst case. Average-case performance seems to be more relevant when comparing these curves: how much time does it cost to scan a query range? I did some rough preliminary experiments for certain ratios between seek time (the cost of ``jumping'' over a curve fragment outside the query range) and scanning time (the cost of scanning a fragment, relative to its length), but I did not find big differences between the curves. Performance varied by at most 10\%, with the simple coil order (Figure~\ref{fig:coil}) giving the best performance, and no apparent advantage for curves with Arrwwid number three. Proper experiments would be needed to verify these observations and assess their validity.

\begin{figure}
\centering
\includegraphics[height=4cm]{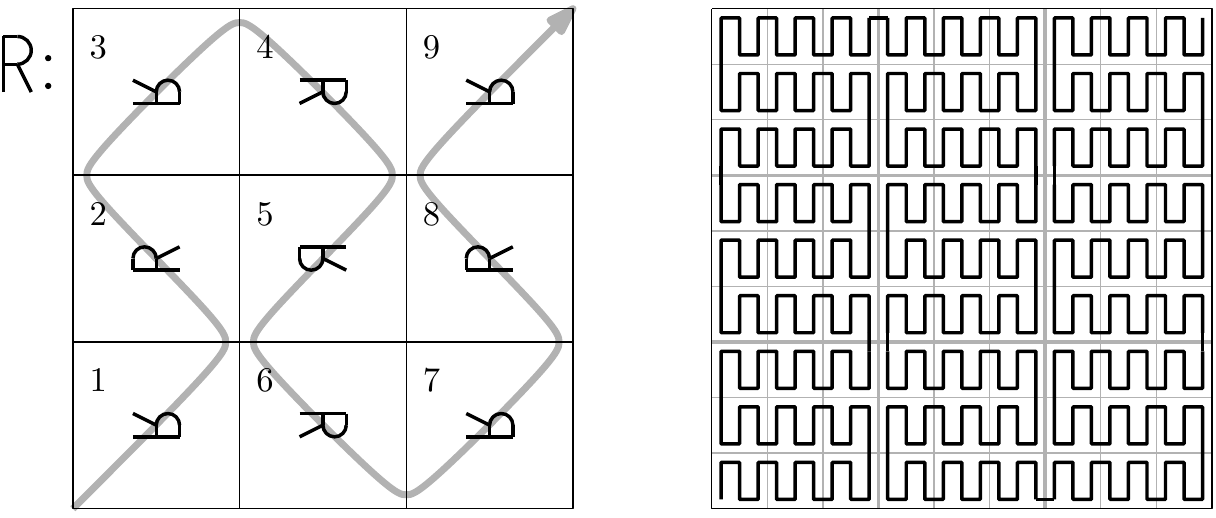}
\caption{Coil order.}
\label{fig:coil}
\end{figure}

In higher dimensions, the stakes are higher. The recursified shifted cube tiling might seem complicated to put to practical use, but one might be able to determine in which tile a point lies by analysing the digits of the coordinates in base-5 notation. This approach may seem interesting to explore because in higher dimensions, there is an exponentially increasing gap in Arrwwid number between regular hypercube tilings and recursified shifted hypercube tilings. On the other hand there is probably also an exponentially increasing gap in cover ratio, in favour of the regular hypercube tilings. Further analysis and possibly experiments would be needed to make a proper comparison.

It is therefore unclear yet whether the research described in this paper led to anything that may be useful in practice. This research was an investigation into where optimisation of the Arrwwid number leads us: What does it take to optimise the Arrwwid number? How limiting is it to consider only regular square or cube curves? What are tilings and curves with low Arrwwid numbers like? I hope that this work opened some new perspectives on how recursive tilings and space-filling curves can be constructed.

\section*{Acknowledgements}
I thank Riko Jacob for inviting me for a short visit to TU Munich, during which discussion with Michael Bader, Riko Jakob and Tobias Scholl about the merits of the \ARRWW curve provided the motivation to start the research described in this paper. I thank Elena Mumford for many inspiring discussions, and in particular for her contribution to the proof of Theorem~\ref{th:3dsfclb}. I found the Kochel curve during a hiking weekend in the Bavarian Alps, while staying in the town of Kochel am See. The Daun tiling was discovered during a hiking trip in the Eifel while staying in the town of Daun.

\bibliographystyle{abbrv}

\begin{thebibliography}{10}

\bibitem{Akiyama2005}
J.~Akiyama, H.~Fukuda, H.~Ito, and G.~Nakamura.
\newblock Infinite series of generalized Gosper space filling curves.
\newblock In {\em China-Japan Conf. on Discrete Geometry, Combinatorics and Graph Theory 2005}, LNCS 4381, pp 1--9, 2007.

\bibitem{Asano1997}
T. Asano, D. Ranjan, T. Roos, E. Welzl, P. Widmayer:
Space-Filling Curves and Their Use in the Design of Geometric Data Structures.
{\em Theoretical Computer Science} 181(1):3--15, 1997.

\bibitem{Bader2006a}%matrix
M. Bader and C. Zenger.
\newblock Cache oblivious matrix multiplication using an element ordering based on a Peano curve.
\newblock {\em Linear Algebra and its Applications} 417:301--313, 2006.

\bibitem{Bugnion1997}
E. Bugnion, T. Roos, R. Wattenhofer, and P. Widmayer.
\newblock Space filling curves versus random walks.
\newblock In M. van Kreveld et al. (eds): {\em Algorithmic Foundations of Geographic Information Systems}, LNCS 1340:199-211, 1997.

\bibitem{Dekking1982}
F. M. Dekking.
\newblock Recurrent sets.
\newblock {\em Advances in Mathematics} 44:78--104, 1982.

\bibitem{Gardner1976}
M.~Gardner.
\newblock Mathematical Games---In which ``monster'' curves force redefinition of the word ``curve''.
\newblock {\em Scientific American}, 235(6):124--133, 1976.

\bibitem{Haverkort2009}
H.~Haverkort and F.~van~Walderveen:
Locality and bounding-box quality of two-dimensional space-filling curves.
{\em Computational Geometry}, 43(2):131--147, 2010.

\bibitem{Hilbert1891}
D. Hilbert.
\newblock \"Uber die stetige Abbildung einer Linie auf ein Fl\"achenst\"uck.
\newblock {\em Mathematische Annalen} 38:459--460, 1891.

\bibitem{Jagadish1990}
H. V. Jagadish.
\newblock Linear clustering of objects with multiple attributes
\newblock {\em ACM SIGMOD Conf. Management of Data (SIGMOD 1990)}, p332-342, 1990.

\bibitem{Jagadish1997}
H. V. Jagadish.
\newblock Analysis of the Hilbert curve for representing two-dimensional space.
\newblock {\em Information Processing Letters} 62:17--22, 1997.

\bibitem{Kumar1994}
A. Kumar.
\newblock A study of spatial clustering techniques.
\newblock {\em 5th Int. Conf. Database and Expert Systems Applications (DEXA 1994)}, LNCS 856:57--71, 1994.

\bibitem{Lebesgue1904}
H. Lebesgue.
\newblock {\em Le\c cons sur l'Int\'egration et la Recherche des Fonctions Primitives,}
\newblock p44--45. Gauthier-Villars, Paris, 1904.

\bibitem{Luxburg1998}
U. von Luxburg.
\newblock {\em Lokalit\"atsma\ss e von Peanokurven.}
\newblock Universit\"at T\"ubingen, Wilhelm-Schickar-Institut f\"ur Informatik, student project report, 1998.

\bibitem{Mandelbrot1983}
B.~B.~Mandelbrot.
\newblock {\em The fractal geometry of nature},
\newblock 1983.% (original edition 1977).

\bibitem{Moon2001}
B. Moon, H. V. Jagadish, C. Faloutsos, and J. H. Saltz.
\newblock Analysis of the clustering properties of the Hilbert space-filling curve.
\newblock {\em IEEE Trans. Knowledge and Data Engineering} 13(1):124--141, 2001.

\bibitem{Niedermeier2002}
R.~Niedermeier, K.~Reinhardt, and P.~Sanders.
\newblock Towards optimal locality in mesh-indexings.
\newblock {\em Discrete Applied Mathematics}, 117:211--237, 2002.

\bibitem{Peano1890}
G.~Peano.
\newblock Sur une courbe, qui remplit toute une aire plane.
\newblock {\em Math. Ann.}, 36(1):157--160, 1890.

\bibitem{Sagan1994}
H.~Sagan:
{\em Space-Filling Curves.}
Universitext series, Springer, 1994.

\bibitem{Scholl2009}
T. Scholl, B. Bauer, B. Gufler, R. Kuntschke, A. Reiser, and A. Kemper.
\newblock Scalable Community-Driven Data Sharing in e-Science Grids.
\newblock {\em Future Generation Computer Systems} 25(3):290--300, 2009.

\bibitem{Teachout2009}
G.~Teachout.
{\em Gary's Fractal Space Filling Curves},
http://teachout1.net/village/fill.html,
retrieved 23 November 2009.

\bibitem{Voorhies1991}
D. Voorhies.
\newblock Space-filling curves and a measure of coherence.
\newblock In: J. Arvo (ed.), {\em Graphics Gems II}, p26--30, Academic Press, 1991.

\bibitem{Wierum2002}
J.-M. Wierum.
\newblock {\em Definition of a new circular space-filling curve:
  {$\beta\Omega$}-indexing.}
\newblock Technical Report TR-001-02, Paderborn Center for Parallel Computing
  (PC$^2$), 2002.

\end{thebibliography}

\end{document}